\documentclass[11pt]{article}
\usepackage{amsmath,amssymb,dsfont,graphicx,cite,rotating,setspace,amsthm,xcolor,bbold}
\usepackage{amsmath,amsthm,amssymb,euscript,verbatim, enumerate}
\usepackage[titletoc,toc,title]{appendix}

\newcommand\Qhat{{\widehat{Q}}}
\newcommand\Phat{\widehat{P}}
\newcommand\qhat{\widehat{q}}
\newcommand{\Ad}{\operatorname{Ad}}
\newcommand{\ad}{\operatorname{ad}}

\newcommand{\beq}{\begin{equation}}
\newcommand{\eeq}{\end{equation}}
\newcommand{\bea}{\begin{eqnarray}}
\newcommand{\eea}{\end{eqnarray}}

\newcommand\vpsi{\boldsymbol{\psi}}
\newcommand\vphi{\boldsymbol{\phi}}
\newcommand{\su}{\mathfrak{su}}

\newcommand{\NLS}{\operatorname{NLS}}
\newcommand{\VFE}{\operatorname{VFE}}

\newcommand\bPsi{{\bf \Psi}}

\newcommand\B{{\bf B}}

\newtheorem{Thm}{Theorem}
\newtheorem{Pro}[Thm]{Proposition}
\newtheorem{Conj}[Thm]{Conjecture}
\newtheorem{Lem}[Thm]{Lemma}
\newtheorem{Cor}[Thm]{Corollary}
\theoremstyle{definition}
\newtheorem{remark}[Thm]{Remark}
\newcommand\T{{\bf{T}}}
\newcommand\N{{\bf{N}}}
\newcommand\V{{\bf{V}}}
\newcommand\W{{\bf{W}}}
\newcommand\kb{{\overline{q}}}
\def\realpart{\operatorname{\sf Re}}
\def\imaginarypart{\operatorname{\sf Im}}
\newcommand{\eq}[2]{\begin{equation}\begin{split}#1\end{split}\label{#2}\end{equation}}

\newcommand\qbar{\overline{q}}
\newcommand\me{\mathsf{e}}
\newcommand\sj{\mathsf{j}}
\renewcommand{\Re}{\operatorname{Re}}
\renewcommand{\Im}{\operatorname{Im}}

\newcommand\vV{\boldsymbol{V}}
\newcommand\vW{\boldsymbol{W}}
\newcommand\vD{\boldsymbol{D}}
\newcommand\vomega{\boldsymbol{\omega}}
\newcommand\vz{\boldsymbol{z}}
\newcommand\vr{\boldsymbol{r}}
\newcommand{\A}{\boldsymbol{\mathcal A}}
\newcommand{\vK}{\boldsymbol{\mathcal K}}
\newcommand\C{\mathbb{C}}
\newcommand{\bbR}{\mathbb{R}}
\newcommand\ZZ{\mathbb{Z}}

\newcommand\D{\mathcal{D}}

\newcommand{\hlambda}{{\widehat{\lambda}}}

\newcommand{\eqnn}[1]{\begin{equation}\begin{split}#1\end{split}\nonumber\end{equation}}

\newcommand{\di}{\partial_x}

\def\Ua{{\bf{U}}}
\def\Um{\widetilde{\bf{U}}}
\newcommand{\ri}{\mathrm{i}}

\newcommand\R{{\mathcal{R}}}

\newcommand{\del}[1]{{\delta_{\bf{V}}{#1}}}

\newcommand{\delg}[1]{\frac{\delta{#1}}{\delta\gamma}}

\newcommand{\delk}[1]{\frac{\delta{#1}}{\delta q}}

\newcommand{\delkb}[1]{\frac{\delta{#1}}{\delta \kb}}

\newcommand{\delkk}[1]{\frac{\delta{#1}}{\delta(q,\overline{q})}}

\newcommand{\AKNS}{\operatorname{AKNS}}

\newcommand\hatpsi{\widehat\vpsi}

\definecolor{forestgreen}{rgb}{0.13, 0.55, 0.13}

\begin{document}
\title
{Stability of closed solutions to the VFE hierarchy with application to the Hirota equation}
\author
{T. Ivey
and S. Lafortune\thanks{Department of Mathematics, College of Charleston, 66 George Street, Charleston, South Carolina, 29424, U.S.
}}

\maketitle

\begin{abstract}
The Vortex Filament Equation (VFE) is part of an integrable hierarchy of filament equations. Several equations part of this hierarchy have been derived to describe
vortex filaments in various situations.
Inspired by these results, we develop a general framework for studying the existence and the linear stability of closed
solutions of  the  VFE hierarchy.  The framework is based on the correspondence
between the VFE and the nonlinear Schr\"odinger (NLS) hierarchies. Our results show that it is possible to establish a connection between
the AKNS Floquet spectrum and the stability properties of the solutions of the filament equations.
We apply our machinery to solutions of the filament equation associated to the Hirota equation. We also discuss how our framework applies to soliton solutions.
\end{abstract}

Mathematics Subject Classification: 35Q55, 37K20,  37K45, 35Q35

\section{Introduction}

Fluid vorticity is often concentrated in small regions. The term {\sl{vortex filament}}
is used to describe the case where vorticity is concentrated in a slender tubular region.
This approximation is important in several problems in fluid mechanics, particularly in superfluidity and turbulence.
For example,  vortex filaments can be used to describe the flow of superfluids very accurately \cite{Sa01}, and
to obtain properties of  the flow in models of turbulent fluids  \cite{Ch91} and superconductors \cite{Crab97}.

A model describing the dynamics of vortex filaments is the Vortex Filament Equation (VFE), also known as the {\sl{Localized Induction Equation}}, given by
\begin{equation}
\label{VFE3}
\gamma_t=\gamma_x \times \gamma_{xx},
\end{equation}
where $t$ is time and $\gamma$ is a curve in $\mathbb{R}^3$ with arclength coordinate $x$ \cite{Ri06}.
(Parametrization by arclength is preserved under the evolution.)
In the derivation of the VFE from the Biot-Savart law, non-local effects (due to, e.g., finite core size, stretching, and interaction with the surrounding fluid) are neglected.
What makes the VFE particularly interesting to study is its connection to integrability, more precisely to the nonlinear Schr\"odinger equation (NLS).  One way of formulating this connection is in terms of Frenet-Serret frame of the curve.

The Frenet-Serret frame associated to the curve $\gamma$ is formed by the unit tangent vector ${\bf{T}}=\gamma_x$, normal vector $\bf{N}$ pointing in the direction of $\gamma_{xx}$, and
binormal vector ${\bf{B}}={\bf{T}}\times {\bf{N}}$.  They satisfy
$$
\begin{aligned}
\T_x&=\kappa\N\\
\N_x&=-\kappa\T+\tau\B\\
\B_x&=-\tau\N,
\end{aligned}
$$
where $\kappa$ and $\tau$ are the curvature and torsion of the curve.
Given a curve that solves the VFE with curvature $\kappa$ and torsion $\tau$,
one can construct the complex quantity
\begin{equation}
q(x,t)={\kappa}{}\exp{\ri \phi}, \quad\text{where }\phi=\int^x\tau\,dx.
\label{hasimoto}
\end{equation}
The time dependence of $\phi$ can be chosen so that
$\phi_t=\frac{\kappa^2}{2}+{\kappa_{xx}}{\kappa}-\tau^2$, and then $q$ solves the NLS
\begin{equation}
\label{NLS}
 q_{t}=\ri q_{xx}+\frac{1}{2}{\ri}|q|^2 q.
\end{equation}
The transformation (\ref{hasimoto}) is known as the Hasimoto transformation \cite{Ha72}.

The NLS equation is Hamiltonian and, in fact, completely integrable.
It can be solved through Inverse Scattering Transform (IST) method
using the associated AKNS linear system given by \cite{AKNS}:
\begin{equation}
\label{AKNS}
\frac{d\vec{\psi}}{dx}=\frac{\ri}{2}\left(
\begin{array}{cc}
	-\lambda &  q \\[8pt]
	 \qbar & \lambda \\
\end{array}\right) \,\vec{\psi}, \;\;\; \frac{d\vec{\psi}}{dt}=
\frac{\ri}{2}\left(
\begin{array}{cc}
	-\lambda^{2}+ \tfrac12|q|^2 & \lambda q+\ri q_{x} \\[8pt]
	\lambda \qbar-\ri \qbar_{x} & \lambda^{2} - \tfrac12|q|^2 \\
\end{array}\right) \,\vec{\psi}.
\end{equation}
The system \eqref{AKNS} is such that the compatibility condition between
the mixed partial derivatives  of $\psi$ forces $q$ to be a solution of NLS \eqref{NLS}.

We can use the AKNS system to invert the Hasimoto map.  To do so, we rewrite the VFE (\ref{VFE3}) as an evolution equation for a curve in $\su(2)$.  Specifically, we define a linear isomorphism $\sj:\su(2) \to \bbR^3$
such that
\begin{equation}
\label{sj}
\sj: \frac{1}{2}\left(
\begin{array}{cc}
-\ri\gamma_0&\gamma_1+\ri\gamma_2\\
-\gamma_1+\ri\gamma_2&\ri\gamma_0
\end{array}
\right) \mapsto (\gamma_0, \gamma_1, \gamma_2).
\end{equation}
This identification takes the Lie bracket to $-1$ times the cross product; thus,
if $\gamma(x,t)$ satisfies \eqref{VFE3} then the matrix $\Gamma(x,t) =\sj^{-1}(\gamma)$ satisfies
\begin{equation}
\nonumber
\Gamma_t=-\left[\Gamma_x,\;\Gamma_{xx}\right].
\end{equation}
This identification enables us to invert the Hasimoto transformation, as follows.
Given a fundamental matrix solution ${\bf{\Psi}}(x,t;\lambda)$ for the AKNS linear system (\ref{AKNS})
such that $\vpsi(0,0;\lambda)=\mathbb{1}$, the matrix-valued function
\begin{equation}
\label{Symf}
\Gamma=\left.{\bf{\Psi}}^{-1}\frac{\partial {\bf{\Psi}}}{\partial \lambda}\right\vert_{\lambda = \lambda_0},
\end{equation}
satisfies
\eq{
\Gamma_t = -\left[\Gamma_x,\;\Gamma_{xx}\right] + 2\lambda_0 \Gamma_x.
}{MVFEg}
(Here and elsewhere we take $\lambda_0$ to be real.)
If we apply $\sj$ to
both sides of \eqref{MVFEg}, then the curve in $\mathbb{R}^3$ satisfies
\eq{
\gamma_t=\gamma_x \times \gamma_{xx}+ 2\lambda_0 \gamma_x.
}{MVFE2}
In particular, $\gamma$ satisfies \eqref{VFE3} when $\lambda_0=0$, and otherwise
is equivalent to a solution of \eqref{VFE3} via  change of variable $x \mapsto x - 2\lambda_0 t$.
The construction \eqref{Symf} is known as the Sym transformation \cite{Sym85}.

An immediate consequence of the connection of the VFE with the NLS equation is
the existence of large classes of special solutions of the VFE (including solitons and multi-solitons, and finite-gap solutions) and the presence of an infinite number of conserved quantities in involution \cite{La91,La94,La99}.  Furthermore, the complete integrability of the VFE makes it the primary model for the study of analytical, geometrical, and topological properties of vortex filaments, and a starting point for deriving more physically realistic evolution equations for  filamentary vortex structures. Many important mechanisms observed experimentally are captured by solutions of the VFE, among them particle transport \cite{ki06} and
propagation of solitary waves in turbulent fluids \cite{Nature82,Hop82,Max83,Max85}.  Vortex solitons in fluids \cite{Konno91,Fukumoto91,Sym85,Levi83,Aref84,Mak03,Kimura04,Miyazaki88},  in superconductors  \cite{Uby95}, and in weather events (see Figure 1 of \cite{Aref84}),  as well as their periodic counterparts---circular vortex rings and knotted vortex filaments  \cite{Ki81,Ke90} in both classical and superfluids \cite{Vel10, SBR98} and laser-matter interactions \cite{Lu96, Mak03}---can be described in terms of exact solutions of the VFE.


Even though the VFE equation appears in several applications, the issue of stability of its solutions was not given due attention prior to work by {ourselves and our collaborators} \cite{Ca11b,Ca11,Lafortune12}. Indeed, stability analysis of VFE solutions can be found in the literature for only a handful of special cases:  a numerical study of the linear stability of steady solutions \cite{Ki82}, a linear stability study  for approximate solutions  in the shape of small-amplitude torus knots  \cite{Ri95},  a linear stability study of helical vortex filaments  \cite{Andersen}, numerical investigations of the stability of Kelvin waves \cite{Sam90,Zhou}, and a study of stability properties of self-similar singular (infinite-energy) solutions  \cite{bv08,bv12}.

The VFE is part of a hierarchy of integrable evolution equations for curves. It is related
to the NLS hierarchy through the Hasimoto transformation \eqref{hasimoto}  \cite{La91,La94,La99}. As such, it shares the same  properties as the NLS hierarchy, that is, it  has an infinite number of conserved quantities which are in involution
with respect to the Poisson bracket  associated to the VFE Hamiltonian \cite{La91}.

Several members of the VFE integrable hierarchy  have been derived as vortex filament models that take into account physical effects neglected in the VFE \cite{Fukumoto91,Fukumoto02,Uby95}.
For example, in an effort at explaining experimental results reported in \cite{Max83,Max85},  Fukumoto and Miyazaki \cite{Fukumoto91}
  derived the following model, which takes into account the presence of a strong  axial flow in the core of the vortex:
  \begin{equation}\vspace{-.1cm}
\label{VFE3e1}
\gamma_t=\gamma_{xxx}+\tfrac{3}{2}\gamma_{xx}\times\left(\gamma_x\times \gamma_{xx}\right)+\sigma(\gamma_x \times \gamma_{xx}),
\end{equation}
where $\sigma$ is a constant. The RHS of \eqref{VFE3e1} happens to be a linear combination of the flow defining the VFE and the flow corresponding to the next member of the hierarchy. Under the Hasimoto map \eqref{hasimoto}, it corresponds to the Hirota equation \cite{Hi73}
\begin{equation}\label{hirotaeqn}
q_t=q_{xxx}+\tfrac{3}{2}|q|^2q_x+\ri\sigma(q_{xx}+\tfrac{1}{2}|q|^2q),
\end{equation}
whose RHS corresponds to a linear combination of the flows defining NLS and the complex mKdV (CmKdV).
Furthermore, the next member of the VFE hierarchy was derived in the context of a vortex filament in a charged fluid on a neutralizing background \cite{Uby95}.
Motivated by these facts, in this article we will develop tools for the study of the existence and stability
of solutions applicable to the whole hierarchy.

The rest of this article is organized as follows. In Section \ref{S2}, we review the basic results \cite{Ca11b,Ca11} we use to study the linear stability of solutions to the VFE.  Section \ref{S3} is devoted to a presentation of the VFE hierarchy of integrable equations.
{{ In Section \ref{S4} we discuss a conjectured result that would enable us to solve the linearized versions of members of the VFE
hierarchy.  While we have not yet proved this result for all members of the hierarchy, in Theorem \ref{T2} we reduce the conjecture to a
condition that can checked via symbolic computation, and the computations we have performed convince us
  that the result is true.}}
Sections \ref{S5} and \ref{S6} are devoted, respectively, to the construction of closed curve solutions of the VFE hierarchy, and the solutions of the corresponding linearized equations. In Section \ref{S7}, we discuss how the squared eigenfunctions are used to solve the linearizations of integrable equation.
{In Section \ref{Hirotasec} our results are specialized to give instability criteria for solutions of the Hirota equation.} Section \ref{Discussion} is devoted to conclusions and a discussion about how our results apply to soliton solutions.
{In Appendix \ref{A}, we derive technical results on the NLS hierarchy needed in the rest of the paper:
the recursive construction of AKNS systems, the effect of gauge transformations on AKNS solutions, and the construction of finite-gap solutions and eigenfunctions for the $m^{\mbox{th}}$-order member of the NLS hierarchy, which we denote by $\NLS_m$.}

\section{Linearized VFE}
\label{S2}

In previous works by the authors and their collaborators \cite{Ca11b,Ca11,Lafortune12}, a method was developed to study the linear stability of solutions to the VFE.
The idea of the method relies on one main result, which we present in this section.

We start by introducing the {\sl{natural frame}} \cite{Ca00, Bishop}, which is particularly appropriate for  the study of the VFE. The natural frame
$\left({\bf{T}},\,{\bf{U}}_1,\,{\bf{U}}_2\right)$ is obtained from the Frenet frame
by rotating the normal and binormal vectors in the normal plane:
$${{\Ua}}_1=\cos \phi \, {\bf{N}} -\sin \phi\, {\bf{B}}, \qquad
{{\Ua}}_2= {\bf{T}}\times {{\Ua}}_1=\sin \phi\, {\bf{N}} + \cos \phi\, {\bf{B}},$$
where $\phi$ is an antiderivative of the torsion (as in the Hasimoto map \eqref{hasimoto}).
The result of this choice of rotation is that the derivatives of
 vectors ${{\Ua}}_1$ and ${{\Ua}}_2$ are multiples of the tangent vector:
 \eqnn{&{\bf{T}}_x=\kappa_1{{\Ua}}_{1}+\kappa_2{{\Ua}}_{2},\\
 &{{\Ua}}_{1x}=-\kappa_1{\bf{T}},\\
  &{{\Ua}}_{2x}=-\kappa_2{\bf{T}},
 }
 where $\kappa_1:= \kappa\cos{\phi}$,  $\kappa_2:= \kappa\sin{\phi}$
 are the {\sl natural curvatures}.
The natural curvatures  thus become the real and imaginary parts of the NLS solution, which has the consequence of simplifying  the Hasimoto map (\ref{hasimoto}) as it now reads
\eq{
q=\kappa_1+\ri \kappa_2.
}{hasimoto2}

The linear stability of a solution $\gamma$ of the VFE is studied by replacing $\gamma$ in \eqref{VFE3} by
 $\gamma+\gamma_1$ and retaining terms up to first order in $\gamma_1$:
 \begin{equation}
 \label{LVFE}
\gamma_{1t}=\gamma_{1x} \times \gamma_{xx}+\gamma_{x} \times \gamma_{1xx}.
 \end{equation}
 Since $x$ in the VFE is the arclength parameter, a VFE solution is required to have a tangent vector with unit length. At first order, this means that $\gamma_1$ must be a locally arclength preserving (LAP) vector field, i.e.~$\gamma_{x}\cdot \gamma_{1x}=0$.

The idea to solve the linearized VFE (\ref{LVFE}) is to look for solutions
of the form of an expansion along the vectors $({\bf{T}},{\bf{U}}_1,{\bf{U}}_2)$, that is
\begin{equation}
\label{ansatz}
\gamma_1=a\,{\bf{T}}+b\,{\bf{U}}_1+c\,{\bf{U}}_2.
\end{equation}
With $\gamma_1$ written in the form above, the LAP condition is realized by the restriction $a_x=b\kappa_1+c \kappa_2$.
We start with a VFE solution $\gamma$  corresponding to an NLS solution $q$  through the Hasimoto transformation (\ref{hasimoto}) (or, equivalently, (\ref{hasimoto2})). We insert the expression (\ref{ansatz}) into the linearized VFE and, by a lengthy but straightforward computation,
we have the following result:
\begin{Thm}[\!\!\cite{Ca11,Ca11b}]
\label{T1}
Let $\gamma_1$ be a solution of the linearized VFE of the form (\ref{ansatz}) and satisfying the LAP condition $\gamma_x\cdot\gamma_{1x}=0$.
Then, the quantity $v:= b+\ri c$  satisfies the linearization of the NLS about the solution $q$.
\end{Thm}
In other words, $b$ and $c$ are the real and imaginary parts of a solution of linearized NLS.
This remarkable observation reduces the problem of solving the linearized VFE to the problem of solving linearized NLS.

The linearized NLS can be solved systematically for large classes of solutions.
Indeed, one can construct solutions of the linearization on the NLS from the solutions of the AKNS system (\ref{AKNS}) through the so-called {\sl{squared eigenfunctions}} \cite{FoLe86, MOv}.
This fact often enables one to solve the linearization completely in the desired
space. For example, in the case where the solution $q$ is a soliton solution (where the space is $L^2\left(\mathbb{R}\right)$),
a complete
set of solutions of the linearization  (in the sense that for any fixed value of time $t$, it spans $L^2(\mathbb{R})$)  can be constructed using the squared eigenfunctions \cite{Kaup76a,Kaup76}.
 In the periodic case,
the question of completeness of the set of squared eigenfunctions remains unsolved. However, it is expected that the same type of results holds \cite{FoLe86}. Note that this question was also considered for  the sine-Gordon equation \cite{ErFoMc87,ErFoMc90}. In the case of some genus one (traveling wave) periodic NLS solution, the question of completeness has been solved recently (see for example \cite{Ivey08,Bottman11}).

Our first goal will be to extend the result of Theorem \ref{T1} to the whole VFE hierarchy.

\section{The VFE hierarchy}
\label{S3}

As a completely integrable system, the NLS has an infinite number of conserved quantities which are in involution
with respect to the Poisson bracket  associated to the NLS Hamiltonian (see for example \cite{La91}).
Each of the equations in the hierarchy is integrable and admits a corresponding AKNS linear system;
the spatial part of each system is the same as in the system \eqref{AKNS} for NLS.  (These AKNS systems
are constructed explicitly in Section \ref{hie}.)
The VFE is also part of a hierarchy of integrable geometric evolution equations for curves
\begin{equation}
\label{filn}
\gamma_t=\W_n, \qquad (\text{VFE}_n)
\end{equation}
for $n \ge 0$,
where $\W_n$ is a vector field along space curve $\gamma$; for example, $\W_0 = -\T$, and $\W_1 = \kappa \B$ gives the VFE.  (Further members of this hierarchy are generated recursively, as explained below.)

The VFE hierarchy is related
to the NLS hierarchy through the Hasimoto transformation \eqref{hasimoto} \cite{La91, La94, La99}.
More precisely, suppose the $m$th order NLS equation is written in evolution form as
\begin{equation}
q_t = F_m[q]\qquad (\text{NLS}_m)
\label{NLSm1}
\end{equation}
for $m\ge 1$. (For example, $F_1[q] = -q_x$, while $F_2[q] = \ri (q_{xx} + \tfrac12 |q|^2 q)$ gives the usual focusing NLS.)
Then, given a solution $\gamma$ of \eqref{filn} with natural curvatures $\kappa_1$ and $\kappa_2$,  the quantity $q=\kappa_1+\ri \kappa_2$ satisfies $q_t = F_{n+1}[q]$.
Conversely, given a solution $q$ of $\NLS_m$ \eqref{NLSm1} and a fundamental matrix solution $\bPsi$ of the corresponding AKNS system, then the Sym transformation \eqref{Symf} evaluated at $\lambda=0$ gives a curve that
(with our identification of $su(2)$ with $\bbR^3$) satisfies $\gamma_t = W_{m-1}$,
and whose natural curvatures are the real and imaginary parts of $q$ (see Prop.~26 in \cite{La99}).
More generally, if \eqref{Symf} is evaluated at a nonzero $\lambda=\lambda_0$, then the resulting
curve will evolve by a linear combination of {$\VFE_{m-1}$ and flows earlier in the hierarchy}. For example, the Sym transformation applied to the NLS equation gives rise to the modified version of the VFE given in \eqref{MVFE2}.  For the generalization to higher order, see Corollary \ref{modCor}.

In this section, we will describe the main properties of the VFE hierarchy.
Each vector field $\W_n$ can be expressed in terms of $\gamma$ and its derivatives with respect to $x$; for example,
\eq{
\begin{aligned}
\W_0&=-\T=-\gamma_x,\\
\W_1&=\kappa\B=\kappa_1\Ua_2-\kappa_2\Ua_1=\gamma_x\times \gamma_{xx},\\
\W_2&=\tfrac12\kappa^2\T+\kappa_x\N+\kappa\tau\B=\tfrac12\kappa^2\T+\kappa_{1x}\Ua_1+\kappa_{2x}\Ua_2=
\gamma_{xxx}+\tfrac{3}{2}|\gamma_{xx}|^2 \gamma_x,\\ 
\W_3&=\kappa^2\tau\T+(2\kappa_x\tau+\kappa\tau_x)\N+(\kappa\tau^2-\kappa_{xx}-\tfrac{1}{2}\kappa^3)\B\\
&=(\kappa_1\kappa_{2x}-\kappa_2\kappa_{1x})\T+\left(\kappa_{2xx}+\frac{\kappa_2(\kappa_1^2+\kappa_2^2)}{2}\right)\Ua_1-\left(\kappa_{1xx}+\frac{\kappa_1(\kappa_2^2+\kappa_1^2)}{2}\right)\Ua_2\\
&=\gamma_{xxxx}\times \gamma_{x}+\gamma_{xx}\times \gamma_{xxx}+{\tfrac{5}{2}} |\gamma_{xx}|^2 \gamma_{xx}\times\gamma_x.
\end{aligned}
}{hiee}

To define the recursion operator that links these vector fields, we first consider a general deformation
$\gamma_t = \W$ of a space curve.  Note that $\W$ is LAP if and only arclength derivative $\W_x$
satisfies $\W_x \cdot \T = 0$.
We define the reparametrization operator $\mathcal{P}$ which takes a arbitrary vector field and redefining its
${\bf{T}}$-component so that the LAP condition is satisfied.
If $\W=f\T+g{\Ua}_1+h{\Ua}_2$ in terms of a natural frame, then
the {\sl reparametrization operator} is
$$
\mathcal{P}(\W)=\int^x(\kappa_1\, g+\kappa_2\,h) dx\, \T+g{\Ua}_1+h{\Ua}_2.
$$
We define the {\sl recursion operator} $\R$ by $\R(\W) = -\mathcal{P}(\T \times \W_x)$
\cite{La91}.  Then the VFE vector fields satisfy the recursion relation
\begin{equation}\label{rec1}
\W_n = \R(\W_{n-1}).
\end{equation}

If we write the vector fields of the VFE hierarchy as
$$\W_n =f_n\T+g_n\Ua_1+h_n\Ua_2,$$
then the $\T$-coefficients $f_n$ (which are given by the integral in the reparametrization operator) are only defined up to an additive constant.  Nevertheless, we will make a `canonical' choice of $\T$-coefficient, expressible in terms of the natural curvatures and their derivatives, given by Langer's formula \eqref{fn}.

If we define the complex quantity $q=\kappa_1+\ri\kappa_2=\kappa e^{\ri\int^x\tau dx}$, then the equation satisfied by $q$ is given by \cite{La94}
\begin{equation}
\label{wave}
q_t=-{\Ua}\cdot \mathcal{R}^2(\W_n),
\end{equation}
 where we have used the following definition
\eqnn{\Ua:=\Ua_1+\ri \Ua_2,}
which satisfies
 \eq{
& \Ua\cdot\Ua=0,\;\;\Ua\cdot\overline{\Ua}=2,\\
&\Ua_x=-q\T,\;\;\T_x=\realpart{{\left(\kb\Ua\right)}}.
 }{Up}
When $n=1$ (i.e., when $\gamma$ evolves by the VFE), equation \eqref{wave} is the {usual focusing} NLS \eqref{NLS}, while for $n=2$
it is the CmKdV \eqref{CmKdV}.  In general, equation \eqref{wave} is the $(n+1)$st order member
of the NLS hierarchy, i.e.,~$\NLS_{n+1}$.


Alternatively, the recurrence relation  operator \eqref{rec1} can be written as \cite{La99}
\begin{equation}
J \W_n=\partial_x\W_{n-1},
\label{rec2}
\end{equation}
with $\W_0\equiv -\T$, and where $J$ denotes the operator $\T\times$. The tangential component (or $\T$-component) of $\W_n$ is determined by the LAP condition, but Langer
determines this tangential component in terms of vector fields earlier in the hierarchy.   For, \eqref{rec2} implies that
$$\W_n = f_n \T - J \di \W_{n-1}$$
for some scalar $f_n$ which must satisfy
$$\di f_n = \di \langle \W_n, \T\rangle = \langle \W_n, \di \T \rangle = -\langle \W_n, J \W_1\rangle = \langle \W_1, J \W_n\rangle.$$
(The LAP condition is used in the second equality.)
Langer observes that for $m<n$
$$\langle \W_m, J \W_n \rangle = \tfrac12 \sum_{j=1}^{n-m} \di \langle \W_{m+j-1}, \W_{n-j}\rangle,$$
so in particular, taking $m=1$, we have
\eq{f_n = \tfrac12 \sum_{j=1}^{n-1} \langle \W_{j}, \W_{n-j}\rangle}{fn}
{for $n\ge 2$}, up to constant of integration.  Throughout this paper, we choose the constant of integration to be zero,
{and note that $f_0=-1$, $f_1=0$ from \eqref{hiee}.}

\section{Linearized equations for the VFE hierarchy}
\label{S4}

For members of the hierarchy other than the VFE, we aim at showing a result similar to Theorem \ref{T1}, that is that the $\Ua_1$ and $\Ua_2$ components of any solution of the  linearization of the
filament equation  \eqref{filn} are the real and imaginary parts of a solution to the linearized corresponding member of the NLS hierarchy.
{For $q$ satisfying the $\NLS_{n+1}$ equation given by \eqref{wave}, the linearization reads}
\begin{equation}
\label{lin}
v_t=-\frac{\delta\beta_{n+2}}{\delta (q,\overline{q})}(v,\overline{v})^T
\end{equation}
where $\delta$ denotes the Fr\'echet derivative and
\begin{equation}
\beta_n\equiv \Ua\cdot \W_n.
\nonumber
\end{equation}
{More precisely, if  $\omega$ is a LAP solution of the linearization of the
filament equation  \eqref{filn}, i.e.
\begin{equation}
\label{linfil}
\omega_t=\frac{\delta \W_n}{\delta \gamma} \omega,
\end{equation}
then we aim to show that the quantity $v=\omega\cdot{{\Ua}}$ solves \eqref{lin}.}
\begin{Lem}
\label{Lem1}
The quantity $v=\omega\cdot{{\Ua}}$ solves (\ref{lin}) for any LAP solution $\omega$ of (\ref{linfil}) if and only if we have the equality
\begin{equation}
\label{cond}
\Ua\cdot\left(\frac{\delta \W_n}{\delta\gamma} {\bf{V}}\right) +{{{\Ua}_{t_n} }}\cdot {\bf{V}}+\frac{\delta\beta_{n+2}}{\delta (q,\overline{q})}\,\left({\Ua} \cdot {\bf{V}},\;\overline{{{\Ua}}}\cdot \overline{{\bf{V}}}\right)^T=0,
\end{equation}
where the derivative of $\Ua$ with respect to the $n^{\mbox{th}}$ flow is given by
\eq{
{\Ua}_{t_n}=-\left(\W_{nx}\cdot {\Ua}\right)\T+\ri\alpha_{n+1}{\Ua},
}{Utn}
and
$\alpha_n =\W_n\cdot\T$ and $\beta_n =\W_n\cdot\Ua$, for any LAP vector field ${\bf{V}}$.
\end{Lem}
\begin{proof}
The equation (\ref{cond}) is obtained by inserting $v=\omega\cdot{{\Ua}}$ into (\ref{lin}), using (\ref{linfil}). The $\T$-component of ${\Ua}_{t_n}$ is found by differentiating the orthogonality condition of $\T$ and $\Ua$ with respect to $t_n$. The fact that ${\Ua_{t_n}}$ has no $\overline{\Ua}$-component is found by differentiating the orthogonality condition of $\Ua$ with itself.  The $\Ua$ component is found by equating the mixed derivatives $\Ua_{xt_n}$ and $\Ua_{t_nx}$ computed with the help of (\ref{Up}).
\end{proof}

As an example, take the VFE with $\W_1=\gamma_x\times\gamma_{xx}=\kappa_1{\Ua}_2-\kappa_2{\Ua}_1$.  We have that $\W_{1x}=-\kappa_{2x}{\Ua}_1+\kappa_{1x}{\Ua}_2$ and thus
$$\W_{1x}\cdot\Ua=\ri q_x.$$
Furthermore,
$$
\begin{aligned}
{\Ua}\cdot \frac{\delta \W_1}{ \delta \gamma}&={\Ua}\cdot\left(\gamma_x\times\partial_x^2-\gamma_{xx}\times\partial_x\right)\\
&={\Ua}\cdot\left(\T\times\partial_x^2-\kappa_1{\Ua}_1\times\partial_x-\kappa_2{\Ua}_2\times\partial_x\right).
\end{aligned}
$$
Thus
$${\Ua}\cdot \frac{\delta \W_1}{ \delta \gamma} ({\bf{V}})={\Ua}\cdot\left(\left(\T\times \partial_x^2\right) {\bf{V}}\right),
$$
for any LAP vector field ${\bf{V}}$, since ${\bf{V}}_x$ has no $\T$ component. A straightforward computation shows that
\begin{align}
\notag
&\R(\W_1)=\frac{q^2}{2} \T+\kappa_{1x}{\Ua}_1+\kappa_{2x}{\Ua}_2,\\ \notag
&\R^2(\W_1)=(\kappa_1\kappa_{2x}-\kappa_2\kappa_{1x})\T+(\kappa_{2xx}+\kappa_2|q|^2/2){\Ua}_1-(\kappa_{1xx}+\kappa_1|q|^2/2){\Ua}_2.
\notag
\end{align}
Hence
$$
\beta_3=\R^2(\W_1)\cdot{\Ua}=-\ri q_{xx}-\ri q |q|^2/2.
$$
We now take the Fr\'echet derivative to find
$$
\frac{\delta\beta_3}{\delta (q,\overline{q})}=(-\ri \partial_x^2-\ri |q|^2,\;-\ri q^2/2).
$$
Finally
$$
\alpha_2=\R(\W_1)\cdot \T=|q|^2/2.
$$
Substituting into the LHS of (\ref{cond}), we get
$$
{\Ua}\cdot\left(\T\times {\bf{V}}_{xx}\right)+\left(-\ri q_x\T+\frac{\ri}{2} |q|^2{\Ua}\right)\cdot {\bf{V}} +\left(-\ri \partial_x^2-\ri |q|^2\right)(\Ua\cdot {\bf{V}})-\ri q^2\overline{\Ua}\cdot \overline{{{\bf{V}}}}/2,
$$
which can immediately can be written as
\begin{equation}
\label{mbz}
{\Ua}\cdot\left(\T\times {\bf{V}}_{xx}\right)-\ri \partial_x^2(\Ua\cdot {\bf{V}})+\left(-\ri q_x\T-\frac{\ri}{2} |q|^2{\Ua}\right)\cdot {\bf{V}}-\frac{\ri q^2}{2} \overline{\Ua}\cdot \overline{{{\bf{V}}}}.
\end{equation}
Let ${\bf{V}}=a\T+b{\Ua}_1+c{\Ua}_2$ with LAP condition given by $a_{x}=\kappa_1 b+\kappa_2 c$. Then
$${\bf{V}}_x=(b_{x}+a\kappa_1){\Ua}_1+(c_{x}+a\kappa_2){\Ua}_2$$ and
$$
{\bf{V}}_{xx}=(b_{xx}+(\kappa_1b+\kappa_2c)\kappa_1+a\kappa_{1x}){\Ua}_1+
(c_{xx}+(\kappa_1 b+\kappa_2c)\kappa_2+a\kappa_{2x}){\Ua}_2+*\T,
$$
where we have intentionally not computed the coefficient of $\T$. A straightforward computation gives us
$$
{\Ua}\cdot\left(\T\times {\bf{V}}_{xx}\right)=\ri v_{xx}+\ri (q\overline{v}+\overline{q}{v})q/2+\ri aq_x,
$$
where $v=b+\ri c$. Substituting into (\ref{mbz}), and using the fact that ${\bf{U}}\cdot{\bf{V}}=v$, we get
$$
\begin{aligned}
\ri v_{xx}+\ri (q\overline{v}+\overline{q}{v})q/2+\ri aq_x-\ri a q_x -\ri v_{xx}-\ri v|q|^2/2-\ri \overline{v} q^2/2=0.
\end{aligned}
$$
Thus the VFE satisfies Lemma \ref{Lem1}.

We now want to prove the following theorem:
\begin{Thm}
\label{T2}
The quantity $v=\omega\cdot{{\Ua}}$ solves (\ref{lin}) for any LAP solution $\omega$ of (\ref{linfil}) if and only if, for any LAP vector field ${\bf{V}}$, we have the equality
\begin{equation}
\label{cond2}
\realpart{\left(
\frac{\delta \alpha_{n+2}}{\delta q}
\left({\Ua} \cdot {\bf{V}}\right)+
\frac{\delta \alpha_{n}}{\delta q}
\left({\Ua} \cdot {\bf{V}_{xx}}+\ri\mu q\right)
-\frac{1}{2}\left(\overline{\beta}_n {\Ua}\cdot {\bf{V}}_x
+\ri\overline{\beta}_{n+1}{{\Ua}} \cdot {\bf{V}}\right)
\right)}=0,
\end{equation}
where $\alpha_n =\W_n\cdot\T$, $\beta_n =\W_n\cdot\Ua$,
and $\mu$ satisfies the equation
\eq{
\mu_x = \imaginarypart{\left(\overline{q}{\bf V}_x \cdot  \Ua\right)}.
}{mu}
\end{Thm}
\begin{proof}
To prove this theorem, we need to show that the LHS of (\ref{cond}) is equal to zero if and only if the LHS of (\ref{cond2}) is also for all $n\geq 0$.
Let $L_n$ denote the LHS of \eqref{cond}, which we rewrite as
\eq{L_n :=\Ua\cdot\del{\W_n}+\Ua_{t_n}\cdot\V+\delkk{\beta_{n+2}}(\Ua\cdot\V,\overline{\Ua}\cdot \V)^T,}{first}
where we use $\del{}$ to denote the Fr\'echet derivative with respect to $\gamma$ in the $\V$-direction, that is
\eq{\del{p} := \displaystyle{\delg{p}(\V)},}{FV}
for any differentiable function $p$ of $\gamma$ and its derivatives.
We are going to prove  that
\eq{
L_{n+1}+\ri \partial_x L_n=-2\ri q T_n,}{finc}
where we define $T_n$ to be the LHS of (\ref{cond2}).
From (\ref{finc}) it is clear that if $L_n=0$ for all $n\geq 0$ then $T_n=0$ for all $n\geq 0$. Now, if $T_n=0$, we use induction and the fact that $L_1=0$ (as verified explicitly for the VFE after the statement of Lemma \ref{Lem1}) to show that $L_n=0$. ($L_0$ can easily be verified to be zero).

To prove (\ref{finc}), we begin by linearizing the recursion relation (\ref{rec2}) by applying the operation $\del{}$:
\begin{equation}
\V_x\times \W_{n+1}+\T\times \del{\W_{n+1}}=\partial_x\del{\W_n},
\label{reclin}
\end{equation}
where we have used the fact that $\del{\T}=\V_x$.
We compute
\eqnn{
\Ua\cdot \del{\W_{n+1}}&=-\ri\Ua\cdot \T\times\del{\W_{n+1}}
\\&=-\ri \Ua\cdot\left(\partial_x\left(\del{\W_n}\right)-\V_x\times \W_{n+1}\right)
\\&=-\ri\left(\partial_x\left(\Ua\cdot\del{\W_n}\right)-\Ua_x\cdot\del{\W_n}\right)+\ri\Ua\cdot \V_x\times \W_{n+1},
}
where on the first line we used the orthonormality of the natural frame $(\T,\Ua_1,\Ua_2)$ and on the second line the relation (\ref{reclin}).  We then substitute for $\Ua\cdot\del{\W_n}$
an expression involving  $L_n$ given in (\ref{first}) to get
\eq{&\Ua\cdot \del{\W_{n+1}}=\\
&\ri\partial_x\left(-L_n+\Ua_{t_n}\cdot\V+\delkk{\beta_{n+2}}(\Ua\cdot\V,\overline{\Ua}\cdot \V)^T\right)-\ri q \T\cdot\del{\W_n}+\ri\Ua\cdot \V_x\times \W_{n+1},
}{indh}
where we have used the expression for $\Ua_x$ given in (\ref{Up}). We now aim at simplifying the expression above by computing $\displaystyle{\delkk{\beta_{n+2}}}$. To do so, we take the dot product on both sides of the recursion relation (\ref{rec2}) with $\Ua$ to obtain
\eqnn{\ri \beta_{n+1}&=\partial_x\left(\Ua\cdot \W_n\right)-\partial_x\Ua\cdot \W_n\\
&=\partial_x\beta_n+q\alpha_n,}
where, on the first line, we have used the orthonormality of the natural frame and, on the second line, we have used the expression for $\Ua_x$ given in (\ref{Up}) and the definitions of $\alpha_n$ and $\beta_n$ given in Theorem \ref{T2}.
We take the Fr\'echet derivative of the above expression with respect to both $q$ and $\overline{q}$ to obtain the following operator equality:
\eq{
\ri\frac{\delta\beta_{n+1}}{\delta q}&=\partial_x\circ \frac{\delta\beta_{n}}{\delta q}+\alpha_n+q\frac{\delta\alpha_{n}}{\delta q},\\
\ri\frac{\delta\beta_{n+1}}{\delta \kb}&=\partial_x\circ \frac{\delta\beta_{n}}{\delta \kb}+q\frac{\delta\alpha_{n}}{\delta \kb}.
}{dis}
By shifting the index $n$ up by 2 in (\ref{dis}) and using  (\ref{indh}), we get the equality
\eq{&\Ua\cdot\del{\W_{n+1}}+\delkk{\beta_{n+3}}(\Ua\cdot\V,\overline{\Ua}\cdot \V)^T+\ri \partial_x f_{n}=\\&\ri\partial_x\left(\Ua_{t_n}\cdot\V\right)-
\ri \alpha_{n+2}(\Ua\cdot\V)-\ri q\delkk{\alpha_{n+2}}(\Ua\cdot\V,\overline{\Ua}\cdot \V)^T-\ri q \T\cdot\del{\W_n}+\ri\Ua\cdot\V_x\times \W_{n+1}.}{tog}
From (\ref{Utn}), we use the recursion relation (\ref{rec2}) to rewrite $\Ua_{t_n}$ as
\eq{
\Ua_{t_n}&=-\left(\Ua\cdot\T\times \W_{n+1}\right)\T+\ri\alpha_{n+1}\Ua\\
&=-\ri \left(\Ua\cdot\W_{n+1}\right)\T+\ri\alpha_{n+1}\Ua,}{Utn2}
where, on the second line, we have used a triple product property and the fact that $\Ua\times \T=\ri\Ua$ from (\ref{Up}).
We now take the dot product with $\V$ on both sides
of the second equality in (\ref{Utn2}) and differentiate with respect to $x$ and obtain
\eq{\ri\partial_x\left(\Ua_{t_n}\cdot\V\right)&
=\left(\Ua\cdot\partial_x \W_{n+1}\right)\left(\T\cdot\V\right)-\alpha_{n+1}\left(\Ua\cdot\V_x\right)
\\&+\frac{q}{2}\left(\beta_{n+1}\left(\overline{\Ua}\cdot\V\right)-\overline{\beta}_{n+1}\left(\Ua\cdot\V\right)\right).
}{UtnVd}
Note that to obtain (\ref{UtnVd}), we used the derivatives given in (\ref{Up}) and the fact that $\W_{n+1}$ is LAP, i.e.~$\partial_x\alpha_{n+1}=\left(\kb\beta_{n+1}+q\overline{\beta}_{n+1}\right)/2$.
Combining  (\ref{tog}), (\ref{Utn2}), and (\ref{UtnVd}) and using definition (\ref{first}), one gets
\eq{
&L_{n+1}+\ri\partial_x L_{n}=\\
&-\ri q\delkk{\alpha_{n+2}}(\Ua\cdot\V,\overline{\Ua}\cdot \V)^T-\ri q \T\cdot\del{\W_n}+\frac{q}{2}\left(\beta_{n+1}\overline{\Ua}-\overline{\beta}_{n+1}\Ua\right)\cdot\V,
}{int}
where we have used the fact that $-\alpha_{n+1}\left(\Ua\cdot \V_x\right)+\ri\Ua\cdot\V_x\times\W_{n+1}=0$ due to the triple product property and the fact that $\Ua\times \T=\ri\Ua$ from (\ref{Up}). We also used $\Ua\cdot\partial_x \W_{n+1}=\ri\Ua\cdot\W_{n+2}$, which follows from (\ref{rec2}) and (\ref{Up}).
Now since $\alpha_n=\T\cdot \W_n$, we have that
\eq{
\T\cdot \del{\W_n}=\del{\alpha_n}-\W_n\cdot \V_x.
}{4}
Using the chain rule and the definition (\ref{FV}), we compute
\eq{
\del{\alpha_n}=\delk{\alpha_n}\circ\del{q}+\delkb{\alpha_n}\circ \del{\kb}.
}{3}
Since $q=\T_x\cdot \Ua$ (from (\ref{Up})),
\eq{
\del{q}&=\V_{xx}\cdot\Ua+\T_x\cdot \del{\Ua}\\
 &=\V_{xx}\cdot \Ua+\T_x\cdot \left(-(\V_x\cdot\Ua)\T+\ri\mu \Ua\right)\\
 &=\V_{xx}\cdot \Ua+\ri\mu q,}{2}
 where $\mu$ satisfies the equation (\ref{mu}) (Note that the $\T$-component of $\del{\Ua}$ is found by applying the operation $\delta$ to the condition $\Ua\cdot \T=0$ and using the fact that $\del{\T}=\V_x$. The fact that $\del{\Ua}$ has no $\overline{\Ua}$-component is found by applying the operation $\del{}$ to the condition of $\Ua\cdot\Ua=0$. The $\Ua$-component is found by equating  the two mixed derivatives $\left(\del{\Ua}\right)_{x}$ and $\del{\left(\Ua_{x}\right)}$ computed with the help of (\ref{Up}).) Finally, by the definition of $\beta_n$ given in Theorem \ref{T2}, we have that
 \eq{
 \W_n\cdot\V_x=\frac{1}{2}\left(\overline{\beta}_n\Ua+\beta_n\overline{\Ua}\right)\cdot\V_x.}{1}
From (\ref{1}), (\ref{2}). (\ref{3}), and (\ref{4}) we find that
\eqnn{
\T\cdot \del{\W_n}=\frac{\delta \alpha_{n}}{\delta q}
\left({\Ua} \cdot {\bf{V_{xx}}}+\ri\mu q\right)+
\frac{\delta \alpha_{n}}{\delta \overline{q}}
\left({\overline{\Ua}} \cdot {\bf{V_{xx}}}-\ri\mu \kb\right)
-\frac{1}{2}\left(\overline{\beta}_n\Ua+\beta_n\overline{\Ua}\right)\cdot\V_x.
}
Inserting the relation above into (\ref{int}), one finds (\ref{finc}).
\end{proof}

While we were not able to prove that the relation (\ref{cond2}) holds for all $n$, this condition is particularly amenable to verification since it only involves the expression of $\W_n$-components in terms of the
natural curvatures. This is to be {{contrasted with}} the first condition found in Lemma \ref{Lem1}, which requires the expression of $\W_n$ as a function of $\gamma$ as well. We were thus able to verify the condition of Theorem \ref{T2} explicitly {(using Maple) up to $n=14$}. We thus make the following conjecture:

\begin{Conj}
\label{conj}
The quantity $v=\omega\cdot{{\Ua}}$ solves (\ref{lin}) for any LAP solution $\omega$ of (\ref{linfil}).
\end{Conj}
In other words, if we have a LAP solution of the form \eqref{ansatz} of the linearization of $\VFE_n$, then $b+\ri c$ satisfies the linearization of $\NLS_{n+1}$. A useful corollary to this conjecture is the fact that this result also applies to linear combinations of different members of the hierarchy:

\begin{Cor}
\label{cor}
If we have a LAP solution of the form \eqref{ansatz} of the linearization of a {linear combination $\sum_{j=0}^n c_j\VFE_j$
for constants $c_j$, then $b+\ri c$ satisfies the linearization of $\sum_{j=0}^n c_j \NLS_{j+1}$}.
\end{Cor}

\section{The General Sym Formula}
\label{S5}

In the introduction, we explained how, given a solution $q(x,t)$ of the NLS equation \eqref{NLS}, and a fundamental matrix solution $\bPsi$ for the AKNS system (such that $\bPsi(0,\lambda)$ is the identity matrix),
the Sym formula \eqref{Symf} at a real value $\lambda=\lambda_0$ gives a curve $\gamma$ which satisfies a modification
of the VFE \eqref{MVFE2}. The modification is adding a constant multiple (proportional to $\lambda_0$) of the unit tangent vector; in particular, when $\lambda_0=0$ we get a solution of the original VFE \eqref{VFE3}.
The fundamental matrix $\bPsi$ used in the Sym formula also provides a moving frame along $\gamma$.  For, if we
fix the following basis for $\su(2)$,
\eq{\me_0 = \dfrac12 \begin{pmatrix} - \ri & 0 \\ 0 & \ri \end{pmatrix}, \quad
\me_1 = \dfrac12\begin{pmatrix} 0 & 1 \\ -1 & 0 \end{pmatrix},\quad
\me_2 = \dfrac12\begin{pmatrix} 0 & \ri \\ \ri & 0 \end{pmatrix},
}{es}
and define $\langle \me_i \rangle = \bPsi^{-1} \me_i \bPsi$, then the spatial part of the AKNS
system
\begin{equation}\label{AKNSmx}
\vpsi_x = U  \vpsi, \qquad U = \dfrac{\ri}2 \begin{pmatrix} -\lambda & q \\ \qbar & \lambda \end{pmatrix},
\end{equation}
implies that
\eqnn{
\langle \me_0 \rangle_x &= \langle [\me_0, U] \rangle = (\Re q) \langle\me_1 \rangle+ (\Im q) \langle \me_2\rangle, \\
\langle \me_1 \rangle_x &= \langle [\me_1, U] \rangle = -(\Re q) \langle \me_0\rangle + \lambda \langle \me_0\rangle, \\
\langle \me_2 \rangle_x &= \langle [\me_2, U] \rangle = -(\Im q) \langle \me_0 \rangle + \lambda \langle \me_1\rangle.
}
If we let
\begin{equation}\label{jangle}
\T = \sj \langle \me_0 \rangle, \qquad \Um_1 = \sj \langle \me_1\rangle,\quad\Um_2 = \sj \langle \me_2 \rangle,
\end{equation}
{using the fundamental matrix evaluated at $\lambda=\lambda_0$}, where $\sj$ is defined in \eqref{sj}, then
\eqnn{&{\T}_x=\kappa_1{{\Um}}_{1}+\kappa_2{{\Um}}_{2},\\
&{{\Um}}_{1x}=-\kappa_1{\T}+\lambda_0\Um_2,\\
 &{{\Um}}_{2x}=-\kappa_2{\T}-\lambda_0 \Um_1,
}where $\kappa_1$ and $\kappa_2$ are the real and imaginary parts, respectively, of the potential $q$ in the
AKNS system.  Again, when $\lambda_0=0$ we obtain a natural frame along $\gamma$ with natural curvatures
$\kappa_1, \kappa_2$; but when $\lambda=\lambda_0$ is nonzero we call $(\T, \Um_1, \Um_2)$ a {\sl twisted
natural frame} along $\gamma$.  (Note that, since the spatial part of the AKNS system is the same
for all evolution equations in the NLS hierarchy, when we begin with a fundamental matrix for
$\AKNS_{n+1}$ this construction yields a natural frame along a solution to $\VFE_n$ when
$\lambda_0=0$.)

In this paper, we are in particular interested in the stability of periodic (i.e., closed) solutions
of the VFE hierarchy.  Intuitively, periodic {$\VFE_n$} solutions should be related to periodic $\NLS_{n+1}$ solutions
by the Hasimoto map, but this is not the case in general.  When the curve $\gamma$ is closed
of length $L$ (i.e., $\gamma(x+L,t) = \gamma(x,t)$)
the potential $q(x)$ associated to $\gamma$ by the Hasimoto map \eqref{hasimoto} is not necessarily periodic; instead
$q(x+L) = e^{\ri a} q(x)$, where $a=\int_0^L \tau(x) \,dx$ is the total torsion.  Similarly, the natural frame
of the curve is not periodic, but satisfies $\Ua(x+L) = e^{\ri a} \Ua(x)$, where $\Ua=\Ua_1 + \ri \Ua_2$.
Similarly, if $q(x,t)$ is a periodic $\NLS_{n+1}$ potential, neither $\gamma$
given by the original Sym formula (with $\lambda=0$) nor the natural frame constructed above are necessarily periodic.

However, it may be possible to choose a nonzero real value {$\lambda_0$} at which the Sym formula yields
a closed curve and such that the moving frame defined by \eqref{jangle} is periodic.  (Conditions specific to the case where $q$ is a periodic finite-gap solutions of the NLS hierarchy are given in Prop.~\ref{closureprop} below.)  In that case,
the curve evolves by the following modified VFE (cf. Corollary \ref{modCor} in Appendix \ref{apgauge}):
\begin{equation}\label{modVFEn}
\gamma_t = \sum_{k=0}^n (-\lambda_0)^k \binom{n+1}{k} \W_{n-k}.
\end{equation}

\section{Linear stability}
\label{S6}
In this section we will relate solutions of the linearized $m$th order NLS equation to solutions of
the linearization of the modified $\VFE_{m-1}$, where the modification is as in \eqref{modVFEn}.
Our aim is to extend an earlier result to higher-order flows in the VFE hierarchy, which we first review.


\subsection{Linearized NLS and VFE}\label{base}
As explained in Section \ref{S5}, evaluating the Sym formula \eqref{Symf} at $\lambda=\lambda_0$ and applying the map $\sj$
defined in \eqref{sj} results in a
an evolving curve in $\bbR^3$ that satisfies the modified VFE
\begin{equation}
\label{MVFE}
\gamma_t = \W_1 -2 \lambda_0 \W_0.
\end{equation}
Using $W_1[\gamma] = \gamma_x \times \gamma_{xx}$ and $W_0 = -\gamma_x$, we obtain the linearization
\begin{equation}\label{MLVFE}
\gamma_{1t}=\gamma_{1x} \times \gamma_{xx}+\gamma_{x} \times \gamma_{xx}+2\lambda_0 \gamma_{x}.
\end{equation}
Here, $\gamma_1$ is a vector field along $\gamma$.  Any such vector field can be expanded in terms of the twisted natural frame obtained as in
\eqref{jangle} from the fundamental matrix used in the Sym formula:
\begin{equation}
\label{ansatzm}
\gamma_1=a\,{\bf{T}}+b\,{{\Um}}_1+c\,{{\Um}}_2.
\end{equation}
When we insert the expression \eqref{ansatzm} into the linearized version \eqref{MLVFE} of the modified VFE, a lengthy but straightforward computation yields the following result:
\begin{Thm}[\!\!\cite{Ca11}]
\label{T2m}
Let $\gamma_1$ be a LAP vector field along a solution of VFE \eqref{MVFE}, related to NLS solution $q$ by the general Sym formula \eqref{Symf}.
Then $\gamma_1$ satisfies linearized modified VFE \eqref{MLVFE} if and only if the quantity $v:= b+\ri c$  satisfies the linearization of the NLS about $q$.
\end{Thm}

The importance of this theorem lies in the fact that potentially relates the {\em periodic} solutions of the linearized modified VFE (mVFE) and linearized NLS. Indeed, assuming we have an $L$-periodic solution of the mVFE \eqref{MVFE2} obtained through the Sym transformation \eqref{Symf}, with periodic twisted natural frame given by \eqref{jangle}, then a periodic solution of the linearized mVFE gives a periodic solution of linearized NLS; the converse is not automatic because the coefficient $a$ in \eqref{ansatzm} might not be periodic.  But, 
as we will see in Prop.~\ref{prop} and Theorem \ref{fund} below, this periodicity is automatic when we use solutions of the linearized NLS that are obtained through squared eigenfunctions.  Thus, if the NLS solution is linearly unstable
with respect to an $L$-periodic perturbation obtained in this way---i.e., if the solution to the linearized
equation has unbounded growth in time---then the corresponding solution of the mVFE  itself is linearly unstable.

\subsection{Extension to the VFE hierarchy}
We now will extend the Theorem \ref{T2m} to the rest of the hierarchy.  In other words, we now have a solution $q$ of $\NLS_{n+1}$
\eqref{NLSm1} and the solution
of the corresponding modified VFE flow \eqref{modVFEn} obtained by applying the Sym transformation \eqref{Symf} evaluated at $\lambda=\lambda_0 \in\bbR$.
As in Section \ref{base}, we have a $\lambda_0$-twisted natural frame $(\T, \Um_1, \Um_2)$ with curvatures $\kappa_1, \kappa_2$
 such that $q=\kappa_1+ \ri \kappa_2$.  We now prove the following theorem:
\begin{Thm}
\label{T3}
We assume that Conjecture \ref{conj} holds.  Let $\gamma_1$ be a LAP vector field along a solution of modified $\VFE_n$ \eqref{modVFEn}, related to $\NLS_{n+1}$ solution $q$ by the general Sym formula \eqref{Symf}.  We expand $\gamma_1$ in terms of
the twisted natural frame as
\begin{equation}
\label{ansatzm3}
{\gamma}_1=a\,{\bf{T}}+{b}\,\widetilde{\bf{U}}_1+{c}\,\widetilde{\bf{U}}_2.
\end{equation}
Then $\gamma_1$ is a solution of the linearization of modified $\VFE_n$ \eqref{modVFEn} if and only if
the quantity $v:= b+\ri c$  satisfies the linearization of the $\NLS_{n+1}$ about $q$.
\end{Thm}
\begin{proof}
{We use Prop.~\ref{Vgaugeprop} from Section \ref{apgauge}, with $m=n+1$, which states} that if $\vpsi$ satisfies {$\AKNS_{n+1}$} as written in \eqref{AKNSm}, then
\eqnn{\hatpsi = \exp( -(\alpha x+(-\alpha)^{n+1}t) \me_0) \vpsi,}
satisfies a modification of \eqref{AKNSm} with $V_{n+1}$ being replaced by a certain linear combination of $V_1,\;V_2,\dots,V_{n+1}$, $\lambda$ being replaced by $\hlambda = \lambda-\alpha$, and the potential $q$ being replaced by
\eq{\qhat := e^{\ri(\alpha x+(-\alpha)^{n+1} t)} q.}{hq}
The evolution equation satisfied by $\qhat$ is (cf.~Prop.~\ref{Vgaugeprop})
\begin{equation}\label{qhatflow}
\qhat_t = \sum_{k=0}^{n} (-\alpha)^k \binom{n+1}k F_{n+1-k}[\qhat],
\end{equation}
{where the $F_{j}[\qhat]$ denote the right-hand sides of the $\NLS_j$ equation, as in \eqref{NLSm1}}.
We choose {$\alpha=\lambda_0$}, and as in the proof of Corollary \ref{modCor}, the curve $\Gamma$ in \eqref{Symf}
is now given by
$$\Gamma = \left.\bPsi^{-1} \dfrac{\partial \bPsi}{\partial \lambda} \right\vert_{\lambda = \lambda_0}
    = \left.\widehat\bPsi^{-1} \dfrac{\partial \widehat\bPsi}{\partial \widehat\lambda} \right\vert_{\widehat\lambda = 0}.
$$
In other words, the same evolving curve arises from the Sym transformation {\em evaluated at $\hlambda=0$} with potential $\qhat$.
{This in turn}
enables us to apply Conjecture \ref{conj} and Corollary \ref{cor}.
{Thus,}  if we expand $\gamma_1$ in terms of an (untwisted) natural frame
\begin{equation}
\nonumber
{\gamma}_1=\hat{a}\,{\bf{T}}+\hat{b}\,{\bf{U}}_1+\hat{c}\,{\bf{U}}_2,  
\end{equation}
then $\hat{b}+\ri\hat{c}$ satisfies the linearization of \eqref{qhatflow}
if and only if $\gamma_1$ satisfies the linearization of modified VFE \eqref{modVFEn}.
The {untwisted} natural frame in question is given by
 \eqnn{
 \T=\widehat{\bPsi}^{-1} \me_0 \widehat{\bPsi}\vert_{\hlambda=0},\;\; \Ua_1=\widehat{\bPsi}^{-1} \me_1 \widehat{\bPsi}\vert_{\hlambda=0},\;\; \Ua_2=\widehat{\bPsi}^{-1} \me_2 \widehat{\bPsi}\vert_{\hlambda=0}.
 }
However, if we expand $\gamma_1$ in terms of the twisted natural frame, as in \eqref{ansatzm3}, then
$b+\ri c = e^{-\ri(\alpha x + (-\alpha)^{n+1} t)} (\hat{b} + \ri \hat{c})$.  Since the change of variables \eqref{hq} takes
solutions of $\NLS_{n+1}$ \eqref{NLSm1} to solutions of the modified $\NLS_{n+1}$ \eqref{qhatflow}, it also takes solutions of the linearization
at $q$ of the first equation to solutions of the linearization of modified $\NLS_{n+1}$ at $\qhat$.  Thus, the expression in \eqref{ansatzm3}
satisfies the {linearized} modified $\VFE_{n}$ if and only if $b+\ri c$ satisfies the linearization of $\NLS_{n+1}$.
\end{proof}

For example, let us take the complex mKdV (CmKdV), which corresponds to the third member of the hierarchy after the NLS
\eq{
q_t=q_{xxx}+\tfrac{3}{2}|q|^2q_x.
}{CmKdV}
Following the construction made in Section \ref{hie}, the coefficient
matrix giving the time dependency in the AKNS system for CmKdV is
\eqnn{
V_3&=\lambda^3P_0-\lambda^2 P_1+\lambda P_2-P_3\\
&= \dfrac{1}{2}\begin{pmatrix} \ri\lambda^3-\frac{\ri}{2}\lambda |q|^2+\frac{1}{2}({\qbar}q_x-q\qbar_x) & -\ri\lambda^2 q+\lambda q_x+\ri(q_{xx}+\frac{1}{2}q|q|^2)\\
-\ri\lambda^2 \qbar-\lambda \qbar_x+\ri(\qbar_{xx}+\frac{1}{2}\qbar|q|^2) &  -\ri\lambda^3+\frac{\ri}{2}\lambda |q|^2-\frac{1}{2}({\qbar}q_x-q\qbar_x) \end{pmatrix}.
}
It can then be computed that the solution obtained from the Sym transformation (\ref{Symf}) evaluated at $\lambda=\lambda_0$ satisfies the equation
\eq{
\gamma_t=\W_2-2\lambda_0\W_1+3\lambda_0^2\W_0,
}{HiFi}
where the $\W$'s are given in (\ref{hiee}). Thus, Theorem \ref{T3} implies that if $b+\ri c$ satisfies the linearized version of the CmKdV, then the LAP vector field given in (\ref{ansatzm3}) satisfies the linearized version of the equation above.

\section{Squared eigenfunctions}
\label{S7}

In this section, we explain the method of squared eigenfunctions, which is used to solve the linearization of  $\NLS_{n+1}$ \eqref{NLSm1},
We first state the result for the NLS and then explains how this generalizes to higher-order members of the hierarchy.
{Then we use this, in conjunction with Theorem \ref{T3}, to give a criterion for linear instability of solutions of modified
$\VFE_n$ \eqref{modVFEn} that arise from finite-gap solutions of $\NLS_{n+1}$.}

For a given NLS potential $q$, we write the linearization of the NLS system  \eqref{NLS} by {substituting} $q\rightarrow q+g$ and retaining terms up to first order in $g$:
\eq{
 g_{t}=\ri g_{xx}+{\tfrac12 \ri q^{2}\overline{g}+\ri |q|^2 g}.
 }{LNLS}

It is well-known \cite{FoLe86,MOv} that solutions of the (\ref{LNLS}) can be constructed from certain quadratic combinations of the solutions of the AKNS system. More precisely, we have the following proposition,
\begin{Pro}
\label{prop}
Given $\vphi$ and $\vpsi$  (not necessarily distinct) solutions  of the AKNS system \eqref{AKNS} at $({q}, \lambda)$,
the following {pairs $(a,g)$ of scalar functions satisfy the system of equations consisting of the linearized NLS \eqref{LNLS} and the equation 
\begin{equation}a_x = \Re( g \overline{q})\label{quasialp}.\end{equation}
\begin{enumerate}
\item[(a)] $a = \Im(\phi_1 \psi_2+ \phi_2\psi_1)$, $g= \phi_1{\psi}_1 +\bar{\phi}_2\bar{\psi}_2$;
\item[(b)] $a= \Re(\phi_1 \psi_2 + \phi_2\psi_1)$, $g=\ri( \phi_1{\psi}_1 -\bar{\phi}_2\bar{\psi}_2)$.
\end{enumerate}
}
\end{Pro}
{The equation \eqref{quasialp} will be used below in constructing LAP vector fields.}

We observe that, while for the soliton case standard inverse scattering techniques can be used to prove that the family of squared eigenfunctions form a complete basis of solutions of the linearized NLS equation in the space $L^2(\mathbb{R})$ for any fixed value of $t$ \cite{Kaup76a,Kaup76}, for periodic boundary conditions the issue of completeness of the squared eigenfunction family is not entirely resolved due to the non-self-adjointness of the spectral problem. Road maps for  completeness proofs for  the periodic sine-Gordon equation (whose spectral problem is also non self-adjoint) are given in \cite{ErFoMc87,ErFoMc90,Krp}. Since we do not have a completeness result in the periodic case, our method applied to closed VFE solutions will yield instability results in the cases where we have periodic solution of the linearized VFE that grow in time, but we will not be able to establish linear stability.

It is believed that Prop.~\ref{prop} generalizes to the whole NLS hierarchy. While there is no proof of this fact,  Prop.~\ref{prop} does extend for a class of solutions that includes the multi-solitons \cite{Yang00,Yang03,Yang2010}. For the example of the Hirota considered in Section \ref{Hirotasec} below, we verify explicitly that  Prop.~\ref{prop} holds.

{Suppose that $q$ is an $L$-periodic solution to $\NLS_{n+1}$, and $\gamma$ is a $L$-periodic (closed) solution to modified $\VFE_n$ \eqref{modVFEn},
associated to $q$ via the Sym transformation \eqref{Symf}.  If eigenfunctions of $q$ can be used to produce, via the quadratic
expressions for $g$ given in Prop.~\ref{prop}, a solution to linearized $\NLS_{n+1}$ that is $L$-periodic then Theorem \ref{T3} lets us construct an $L$-periodic solution $\gamma_1$ to linearized modified $\VFE_n$.
(When $a$ is given by the corresponding formula in Prop.~\ref{prop}, then \eqref{quasialp} ensures that $\gamma_1$ is LAP.)
Moreover, if the eigenfunctions have exponential growth in time, then we conclude that $\gamma$ is linearly unstable.
If $q$ is a finite-gap solution, the solutions of its AKNS system are given
by the Baker eigenfunctions \eqref{Baker}, and the quadratic expressions have period $L$ if $\lambda$ is a periodic or antiperiodic point
of the Floquet spectrum of $q$ (see Remark \ref{quasiclosure} for definitions).  Moreover, the eigenfunction formulas \eqref{Baker} show that these have exponential growth if and only if the Abelian integral $\Omega_{n+1}$ has nonzero imaginary part.}

{In fact, more is true. As discussed at the end of Section \ref{S5}, closed curves in general correspond under the Hasimoto transformation to potentials $q$ that are only quasiperiodic, since the normal vector $\bf{U} = \bf{U}_1 + \ri \bf{U}_2$ is itself only quasiperiodic.
Conversely, by Remark \ref{quasiclosure} an $L$-quasiperiodic finite-gap solution to $\NLS_{n+1}$ gives rise to a closed filament $\gamma$ when
the Sym formula \eqref{Symf} is evaluated at a point $\lambda_0$ satisfying the closure conditions of Prop.~\ref{prop2}.
Using eigenfunctions at a point $\lambda$ belonging to the periodic spectrum of $q$ (again, see Remark \ref{quasiclosure} for
what this means in the quasiperiodic case) yields a solution $g$ of linearized $\NLS_{n+1}$ that has the same quasiperiodic behavior as $q$ given by \eqref{finitegapq}, i.e.,
\begin{equation}\label{quasiv}
g(x+L,t) = e^{-iEL} g(x,t)
\end{equation}
(compare with equation \eqref{quasiq}).  On the other hand, a fundamental matrix $\Psi$ comprised of Baker eigenfunctions evaluated at $\lambda_0$ satisfies
$\bPsi(x+L,t) = \exp( EL e_0) \bPsi(x,t)$ (see \eqref{Baker}), so the twisted natural frame defined by \eqref{jangle} satisfies
$\widetilde{\bf{U}}(x+L,t) = e^{-EL} \widetilde{\bf{U}}(x,t)$.  It then follows easily from \eqref{quasiv} and \eqref{quasialp} that $\gamma_1$, as defined by \eqref{ansatzm3}
with $b+\ri c = g$ and $a$ given by the appropriate formulas in Prop.~\ref{prop}, is periodic and LAP.  Again, if $g$ has exponential growth in time, then $\gamma$ is linearly unstable.}

{We summarize these discussions in the following theorem:}

\begin{Thm}
\label{fund}
{Let $\gamma$ be an $L$-periodic solution $\gamma$ of modified $\VFE_n$ \eqref{filn} corresponding to a {quasiperiodic finite-gap} solution $q$ of $\NLS_{n+1}$ \eqref{NLSm1} through the Sym {formula} \eqref{Symf}.}  If the Floquet spectrum of $q$ includes an $L$-periodic or $L$-antiperiodic
{point} at which $\Im{\Omega_{n+1}(P)}\neq 0$, then $\gamma$ is linearly unstable.
\end{Thm}

\section{Examples}
\label{Hirotasec}

This section is devoted to examples. We will consider solutions of the CmKdV equation \eqref{CmKdV}, which 
is a particular case of the Hirota equation \eqref{hirotaeqn}.
From (\ref{MVFE2}) and (\ref{HiFi}) we find that under the Sym transformation (\ref{Symf}), the CmKdV equation corresponds to 
\eq{
\gamma_t&=\W_2-2\lambda_0\W_1+3\lambda_0^2\W_0\\
&=\gamma_{xxx}+\tfrac{3}{2}|\gamma_{xx}|^2 \gamma_x-2\lambda_0(\gamma_x\times \gamma_{xx})+3\lambda_0^2\gamma_x.
}{modCMK}
where the $\W$'s are given in (\ref{hiee}).

In order to obtain examples, we aim at choosing periodic solutions of the CmKdV whose Floquet spectrum contains a value $\lambda=\lambda_0$ at which the Sym transformation \eqref{Symf} yields a periodic curve, i.e.~a value $\lambda=\lambda_0$ satisfying the closure conditions specified by Prop.~\ref{closureprop} below.  This is done using the construction described in Section \ref{Apfg} by choosing the $2g+2$ branch points $\lambda_i$ ($g$ is the genus of the solution) for which the potential $q$ given in \eqref{finitegapq} is {quasiperiodic} with Floquet spectrum containing an appropriate value $\lambda=\lambda_0$.

\subsubsection*{Genus One Examples}

Finite-gap solutions of the NLS hierarchy generated by genus one Riemann surfaces correspond under the Hasimoto map to initial data that are elastic rod centerlines.  Langer \cite{La99} points out that along such curves all VFE vector fields are {\em slide-Killing}, i.e., they differ from the restriction of a Euclidean Killing
field by a constant multiple of the unit tangent vector.  In particular, the shape of such curves is preserved
exactly under any linear combination of flows in the VFE hierarchy.

Using explicit formulas for the Floquet discriminant available for genus one solutions (see, e.g., \cite[\S3.5]{CI05})
it is relatively easy to determine sets of branch points for which the closure conditions in Prop.~\ref{prop2}
are satisfied.  For example, the branch points $\lambda_1 = -0.674894 +1.23316 \ri$,
$\lambda_2 = 0.674894 + 0.242667 \ri$ (along with their complex conjugates) satisfy these conditions at $\lambda_0 = -.664983$.  Using this
data to produce a genus one solution of the CmKdV equation \eqref{CmKdV}, and using this value in
the Sym formula produces an elastic rod in the form of a nearly planar trefoil knot, whose motion under \eqref{modCMK} gives relatively large weight
to vector field $\W_1$, i.e., equation \eqref{modCMK} with $\lambda_0 = -.664983$ gives
\eq{
\gamma_t=\W_2+1.320 \W_1+1.327 \W_0.
}{mVFE2}
This generates motion along an approximately spiral path that is elongated in the
direction of the axis of the trefoil (see Figure \ref{trefoilone}).  The Floquet spectrum for this potential appears in Figure \ref{trefspectrum},
in which the spines of the continuous spectrum terminate at the branch points, and each of the left-hand
spines contains a complex double point $\lambda_d$.  We compute that $\Omega_3(\lambda_d)$ has a nonzero
imaginary part, thus indicating that this solution of \eqref{modCMK} is linearly unstable.

\begin{figure}[hb]
\centering

\includegraphics[width=5in]{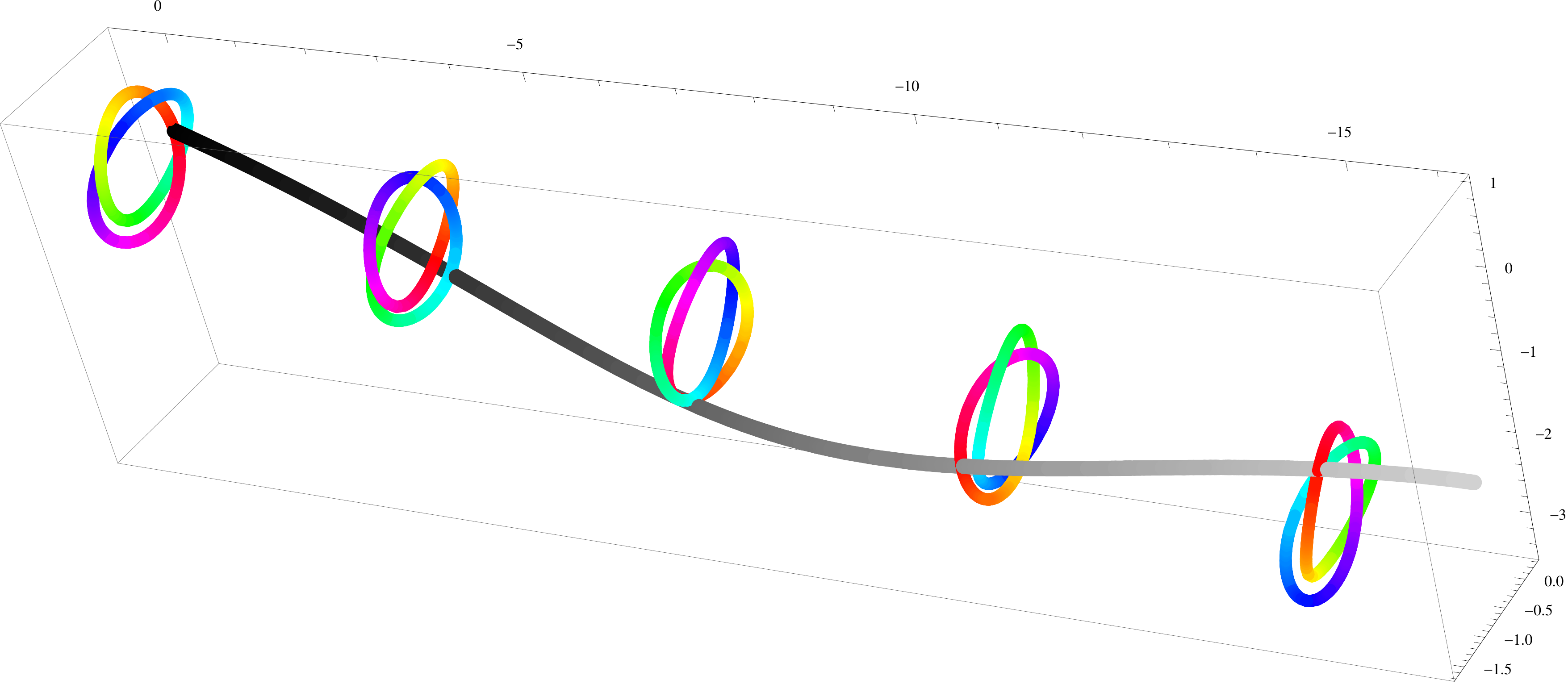}

\caption{A genus one solution of the {modified $\VFE_2$ flow} {\eqref{mVFE2}}.  The black curve
traces the evolution of a fixed point on the filament, fading to gray in the direction
of increasing time.  The shape of the filament is preserved because the VFE vector fields
restrict to genus one filaments to be slide-Killing.}
\label{trefoilone}
\end{figure}

\begin{figure}[hb]
\centering
\includegraphics[width=4in]{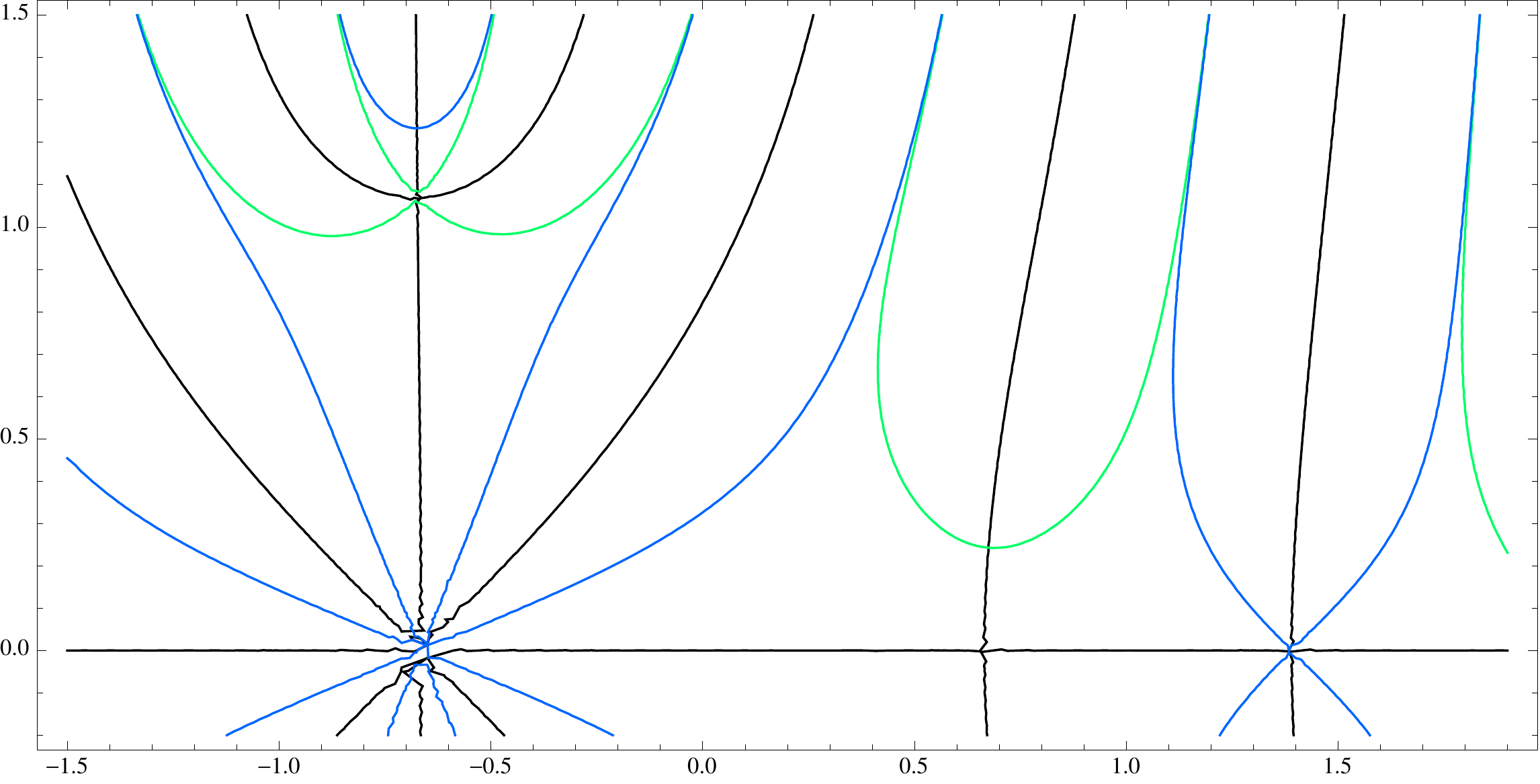}

\caption{Level set diagram for the Floquet discriminant $\Delta(\lambda)$ of the AKNS system associated
to the solution shown in Figure \ref{trefoilone}; see \cite[\S1.1]{CI07} for more details.
The spectrum is symmetric
under complex conjugation, so only the upper half of the $\lambda$-plane is shown.
Black curves are where $\Im \Delta = 0$,
while green and blue curves are where $\Re \Delta =$ 2 or -2, respectively.
The continuous spectrum consists of the real axis, and the parts of the black curves
between the real axis and the branch points, which are at simple intersections of blue and green curves with the black spines. Multiple periodic points appear along the real axis, and near the end of the spine on the left.}
\label{trefspectrum}
\end{figure}

{As discussed in Remark \ref{quasiclosure}}, the closure conditions are preserved when the branch points are translated to the right or left
by adding a real constant.  In particular, if we use the branch points $\tilde\lambda_1 = \lambda_1 - \lambda_0$ and
$\tilde\lambda_2 = \lambda_2 - \lambda_0$, then the closure conditions are satisfied at $\lambda=0$.  Thus,
evaluating the Sym formula at the origin in this case produces a trefoil knot that moves purely by the
{$\VFE_2$ flow} $\gamma_t = \W_2$.  As shown in Figure \ref{trefoiltwo}, under this flow the trefoil undergoes
many more oscillations along its plane, and has a relatively small component of motion perpendicular to its plane.
This makes sense, since $\W_2 = \tfrac12 \kappa^2 \T + \kappa' \N + \kappa\tau \B$ preserves
planarity, and has a small transverse component when the torsion is small.
Since the Floquet spectrum of this solution is merely translated to the right, the value of
$\Omega_3(\lambda_d)$ and the solution is again linearly unstable.

\begin{figure}[hb]
\centering

\includegraphics[width=4in]{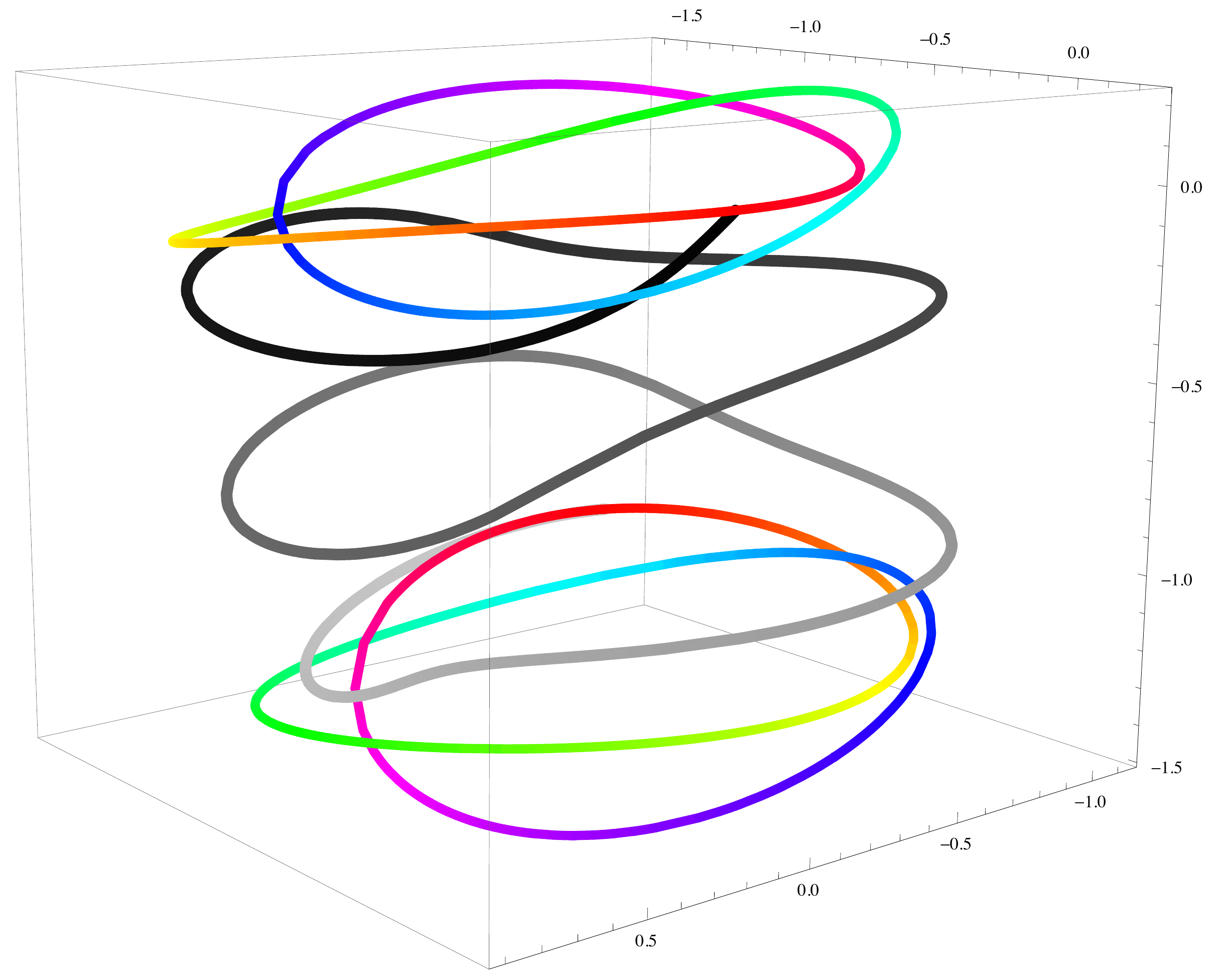}

\caption{A genus one solution of the (unmodified) $\VFE_2$ flow.}
\label{trefoiltwo}
\end{figure}

\subsection*{Higher Genus Examples}
It is difficult, in general, to find configurations of branch points for higher-genus Riemann surfaces
that satisfy the closure conditions. However, closure conditions can be preserved by isoperiodic
deformations (see \cite{CI07}) of the branch points, which can also be used to increase the genus
by opening up double periodic points along the real axis.

For our third example, the branch point configuration is obtained by starting with the Floquet
spectrum of the double cover of the genus one planar figure-8 elastic rod, and symmetrically opening up two real double points using isoperiodic deformations.
(These spectra, which are also symmetric about the origin, are shown in Figure \ref{beforeafter}; closure
conditions are satisfied at $\lambda=0$ before and after the deformation.)
The resulting branch points are
$\lambda_1 =  -2.2271 + 0.8015 \ri$, $\lambda_2 =  -0.8390 + 1.974\ri$, together with their reflections
in the origin and the real and imaginary axes.
(This same configuration of branch points can be used to generate a solution
which maintains a plane of reflection symmetry and is periodically planar (see Figure 8 in \cite{CI05}),
provided that one chooses $\vD = (0,0,0)$ in the finite-gap formulas \eqref{Baker} and \eqref{finitegapq}; in this case,
we break symmetry by using $\vD = (-1, .2, -.1)$.)

Reconstructing at the origin gives a closed genus 3 solution of the $\W_2$-flow, shown at three successive times
in Figure \ref{highergone}.  (Note that this curve initially has knot type $6_3$, is unknotted in the
next snapshot, then returns to its original knot type.)
The continuous spectrum contains four complex double points,
located $\pm .79065 \pm 1792\ri$, one along each of the spines extending from the origin; evaluating $\Omega_3$ at these points indicates that this solution is again linearly unstable.

\begin{figure}[hb]
\centering

\includegraphics[width=4in]{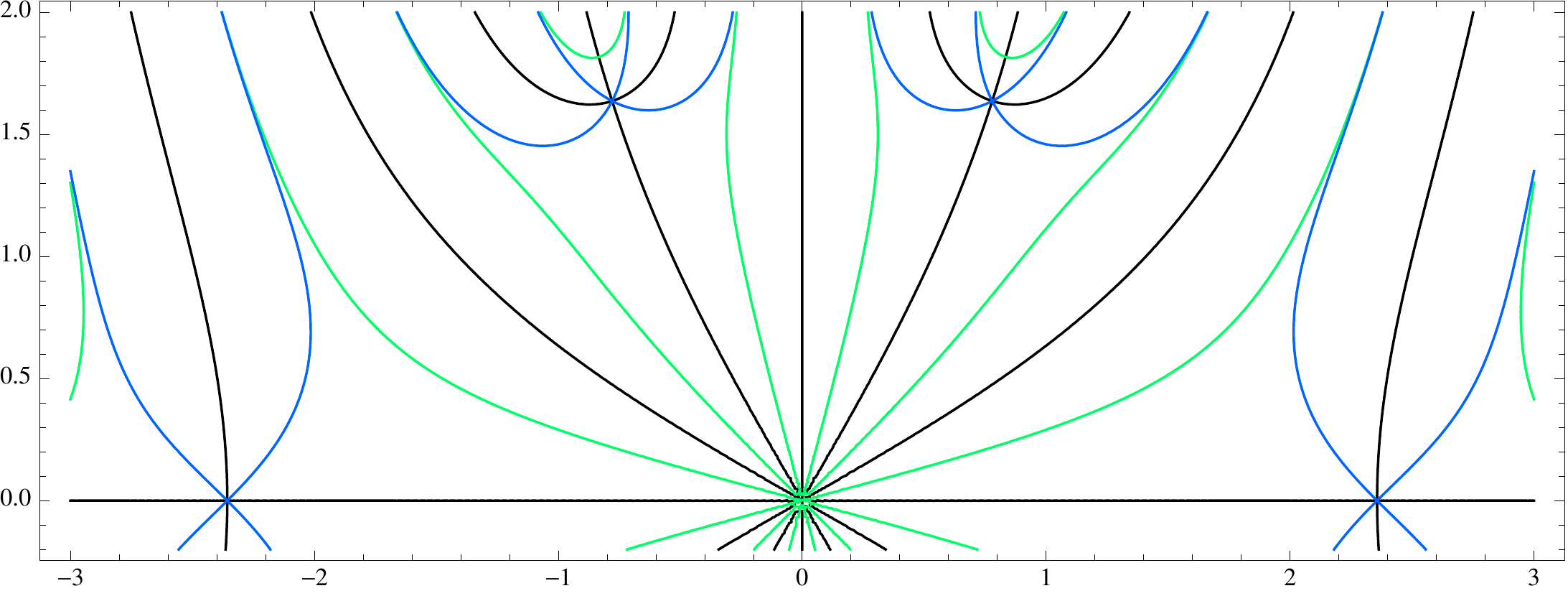}

\includegraphics[width=4in]{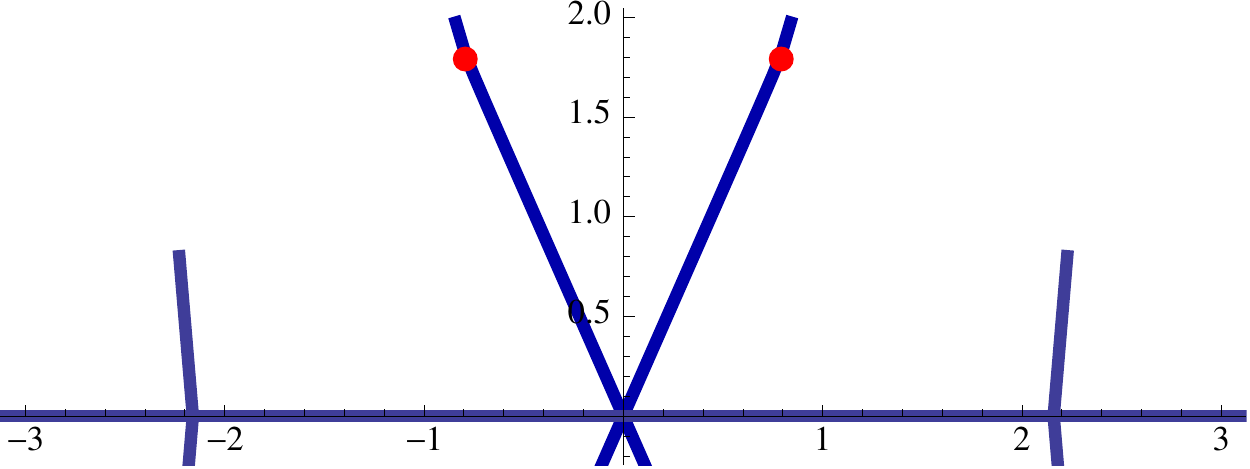}
\caption{The upper plot shows the level curve diagram for the Floquet spectrum of a double-covered figure 8 elastic rod; the green, blue and black curves have the same meaning as in Figure \ref{trefspectrum}.  Two spines emerge from the origin into the upper half-plane, each carrying a double periodic point; additional double points appear along the real axis.
The lower plot shows the continuous spectrum after isoperiodic deformations have been used to symmetrically opening up the real double
points nearest the origin; the complex double points are indicated by red dots.}
\label{beforeafter}
\end{figure}

\begin{figure}[hb]

\includegraphics[width=4in]{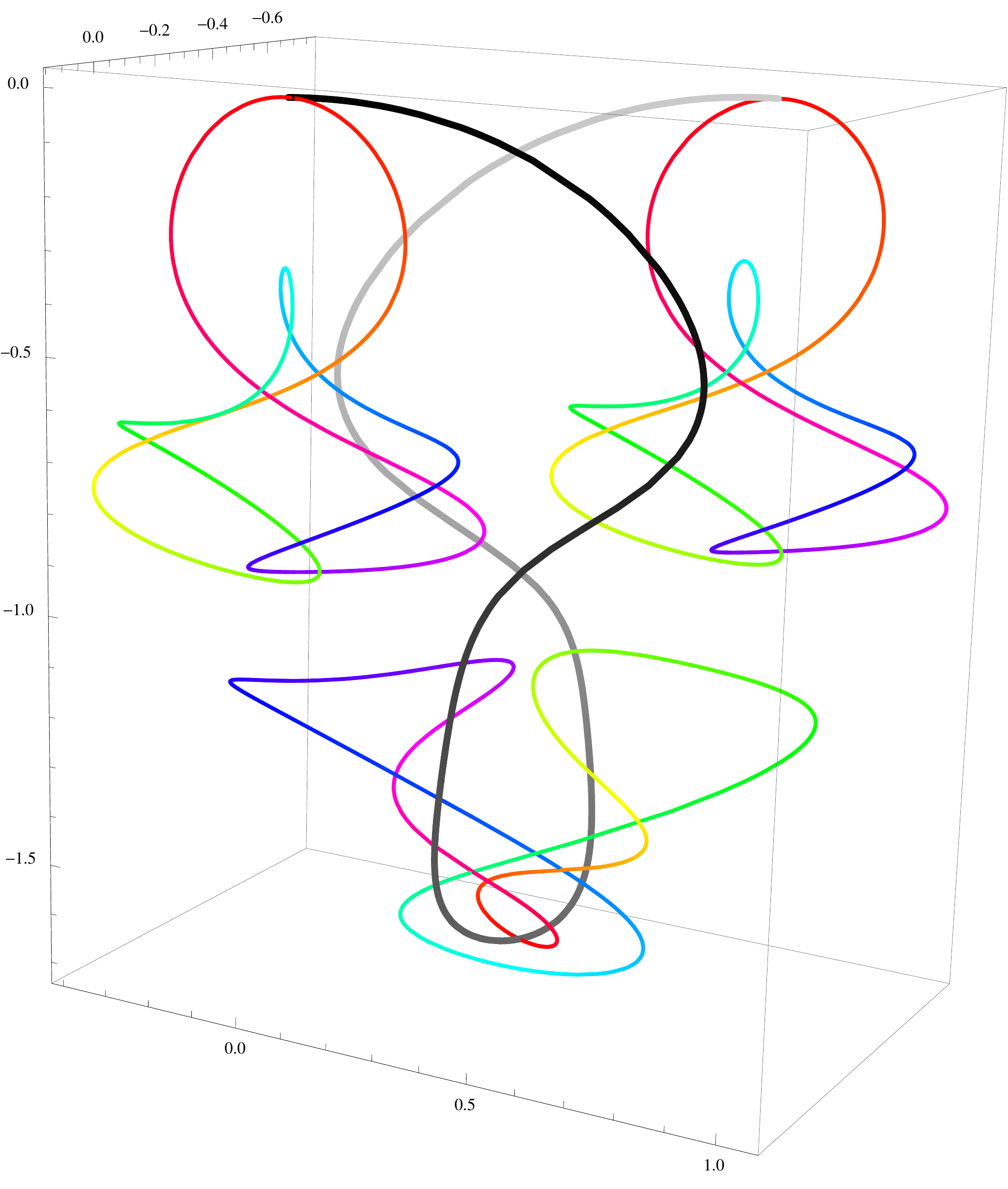}
\centering
\caption{A genus-3 solution to the $\VFE_2$ flow.}
\label{highergone}
\end{figure}

\section{Discussion and Conclusions}
\label{Discussion}

In this paper, we introduced a method to study the linear stability of solutions to the VFE hierarchy. The main result given in Conjecture \ref{conj}  enables us to find linear instabilities
of closed solutions as it provides a way to construct the solutions of the linearized $\VFE_{n}$ in terms of the linearized $\NLS_{n+1}$.
{This conjecture has been verified by symbolic computations up to $n=14$.}
{Since solutions
of the linearized version of members of the NLS hierarchy can be constructed through the squared eigenfunctions,
the same can be said of the linearized versions of the
equations of the VFE hierarchy---that is, solutions of the linearized version of the  $\VFE_n$ (or a modified version of it) can be obtained through the squared eigenfunctions associated to $\NLS_{n+1}$.}

{Our general scheme, then, is to consider a periodic $\NLS_{n+1}$ solution that is associated to
a closed (periodic) curve solution of a modified version of $\VFE_n$.  (Again, the term `modified' refers to
the fact that the Sym transformation takes a solution of the $\NLS_{n+1}$ to a curve satisfying $\VFE_n$ modified by the addition of a linear combination of lower order members of the hierarchy, as in \eqref{modVFEn}.)
If one finds a solution of the linearized version of $\NLS_{n+1}$, with period equal to the
length of the space curve and which exhibits exponential growth in time, then the closed filament is
shown to be linearly unstable.}
This is in essence what Theorem \ref{fund} says, as it relates instabilities of closed solutions of a modified $\VFE_n$ to the Floquet spectrum of the corresponding $\NLS_{n+1}$ solution.

While we have focussed our discussion on closed solutions, our method can be used for any type of solutions. The method will be most useful when a large set of solutions of the linearized version of the $\NLS_{n+1}$  is known. In the case of the (multi-)soliton solutions, the set of solutions of the linearization of $\NLS_{n+1}$ about a soliton solution obtained through the squared eigenfunctions is complete in $L^2(\mathbb{R})$ for any fixed value of $t$\cite{Kaup76,Kaup76a,Kaup90,Yang00,Yang03,Yang2010}. {Having a complete set may enable one to establish stability in a certain sense.  For example, in the case of one-soliton solutions the squared eigenfunctions} can be used to show that the point spectrum of the eigenvalue problem arising from the linearization is restricted to the imaginary axis \cite{Kaup90,Yang00}, thus implying the spectral stability for the corresponding one-soliton solutions of the VFE hierarchy through Conjecture \ref{conj}.

\subsubsection*{Acknowledgements}

The authors are grateful to Annalisa Calini and Alex Kasman for useful discussions. This work was  supported by the National Science Foundation through grant DMS-0908074 (S. Lafortune).


\begin{appendices}
  \renewcommand\thetable{\thesection\arabic{table}}
  \renewcommand\thefigure{\thesection\arabic{figure}}
\section{The NLS Hierarchy}
\label{A}
\subsection{Construction}
\label{hie}

We begin with the AKNS system for $m$th flow of the NLS
hierarchy,
\begin{equation}\label{AKNSm}
\vpsi_x = U\vpsi, \qquad \vpsi_t = V_m \vpsi
\end{equation}
where the first equation is the same as the first one in \eqref{AKNS}.
Following Fordy \cite{Fordy}, we assume that $V_m$ is a polynomial of degree $m$ in $\lambda$:
\begin{equation}\label{Vmex}
V_m = \sum_{j=0}^m (-\lambda)^{m-j} P_j,
\end{equation}
where the $P_j(x,t)$ take value in $\su(2)$.
(Comparing with \eqref{AKNS} shows that $m=2$ corresponds to the usual NLS.)
The compatibility condition for \eqref{AKNSm} is
\begin{equation}\label{compatt}
U_t =  [V_m , U] + \di V_m.
\end{equation}
For purposes of expanding this in terms of $\lambda$, we write
\begin{equation}\label{defofU}
U = \lambda \me_0 + Q, \qquad Q := \dfrac{\ri}2 \begin{pmatrix} 0 & q \\ \qbar & 0 \end{pmatrix},
\end{equation}
with $\me_0$ given in \eqref{es}. Substituting this and the expansion for $V_m$ into \eqref{compatt}, we get
\begin{align}
Q_t &= [P_m, Q] + \di P_m, \label{Qtime} \\
0 &= [P_j, Q] - [P_{j+1}, \me_0] + \di P_j, \qquad 0 \le j < m, \label{mittelj} \\
0 &= [P_0, \me_0], \notag
\end{align}
from the coefficients of $\lambda$ to the power $0$, $m-j$ and $m+1$ respectively.
The last equation implies that $P_0$ is a multiple of $\me_0$, and taking the $\me_0$ component
of \eqref{mittelj} with $j=0$ shows that the multiple must be constant in $x$.  We will take the {\it ansatz} that $P_0= \alpha \me_0$
for a real constant $\alpha$.  (This will be chosen later, in order to make $V_2$ match the usual AKNS
system for NLS.)  Similarly, each $P_{j+1}$ is determined by \eqref{mittelj} up to an additive constant
times $\me_0$.

Notice that \eqref{mittelj} gives a recursively-defined infinite sequence $\{ P_j\}$.  Langer \cite[\S5.2]{La99} observes that
when we set $q=\kappa_1 + \ri \kappa_2$,
the components of $P_j$ with respect to the $\me_a$ satisfy the same recursion relations
as the components of $\W_j$ with respect to the natural frame.  In more detail, if we write
$$\W_n = f_n \T + g_n \Ua_1 + h_n \Ua_2,$$
then the recursion relation \eqref{rec2} $J \W_{n+1}= \di \W_n$ gives
$$g_{n+1} = \kappa_2 f_n + \di h_n, \qquad -h_{n+1} = \kappa_1 f_n + \di g_n,$$
while the LAP condition gives
$$\di f_n = \langle \W_n, \di \T \rangle = \kappa_1 g_n + \kappa_2 h_n.$$
On the other hand, taking $P_n= f_n \me_0 + g_n \me_1 + h_n \me_2$ gives
\eqnn{
0 = [P_n, Q] + [\me_0, P_{n+1}] + \di P_n& = [ f_n \me_0 + g_n \me_1 + h_n \me_2, \kappa_1 \me_2 - \kappa_2 \me_1]
+ [\me_0,  f_{n+1} \me_0 + g_{n+1} \me_1 + h_{n+1} \me_2]  \\
&+ (\kappa_1 g_n + \kappa_2 h_n) \me_0 +(\di g_n) \me_1 + (\di h_n) \me_2\\
=
(&\kappa_1 f_n + h_{n+1} + \di g_n)  \me_1 + (\kappa_2 f_n - g_{n+1} + \di h_n)\me_2
.}

Thus, we will solve the NLS recursion \eqref{mittelj} by taking
$$P_n = \alpha ( f_n \me_0 + g_n \me_1 + h_n \me_2).$$
For example, by expressing the first few vector fields $\W_n$ as given in \eqref{hiee} in terms of the natural frame, we obtain
\begin{align*}
P_0 &= -\alpha \me_0, \\
P_1 &= \alpha(\kappa_1 \me_2 - \kappa_2 \me_1) = \alpha Q,\\
P_2 &= \alpha(\kappa_1' \me_1 + \kappa_2' \me_2+ \tfrac12(\kappa_1^2 + \kappa_2^2) \me_0)
= \dfrac{\alpha}2 \begin{pmatrix} -\tfrac12 \ri |q|^2  & q' \\ -\qbar' & \tfrac12 \ri |q|^2  \end{pmatrix}.
\end{align*}
From \eqref{Qtime} we get the induced evolution equations for scalar $q$:
\begin{itemize}
\item For $m=0$, we get $\ri q_t = -\ri \alpha q$, an infinitesimal symmetry of NLS.
\item For $m=1$, we get $q_t = \alpha q_x$, another NLS symmetry.
\item For $m=2$, we get $\ri q_t = \alpha (q_{xx} + \tfrac12 |q|^2 q)$.
\end{itemize}
So we will choose $\alpha=-1$ to obtain the usual (focusing) NLS.  In particular, the coefficient
matrix giving the time dependency in the AKNS system for NLS is
$$
\begin{aligned}
V_2 &= \lambda^2 P_0 - \lambda P_1 + P_2 \\
      &= \lambda^2 \me_0 + \lambda Q  - \kappa_1' \me_1 - \kappa_2' \me_2 -\tfrac12 (\kappa_1^2 + \kappa_2^2) \me_0\\
	&= \dfrac{1}{2}\begin{pmatrix} \ri(-\lambda^2+\tfrac12|q|^2) & \ri \lambda q - q' \\
					 \ri\lambda \qbar + \qbar' & \ri(\lambda^2 - \tfrac12|q|^2) \end{pmatrix}.
\end{aligned}
$$

The benefit of constructing the NLS hierarchy in a way that is closely linked with the recursive construction
of the VFE hierarchy is that the Sym transformation automatically generalizes to higher-order flows {(cf. Prop.~26 in \cite{La99})}:

\begin{Pro}\label{allSym0} Let $\Psi(x,t,\lambda)$ be a matrix solution, assumed to be $SU(2)$-valued when $\lambda$ is real,
of the $m$th AKNS system \eqref{AKNSm} at $q(x,t)$.
Then the Sym formula at $\lambda = 0$, i.e.,
\begin{equation}\label{Sym0}
\Gamma = \Psi^{-1}\left. \dfrac{\partial \Psi}{\partial \lambda}\right\vert_{\lambda=0}
\end{equation}
gives a solution of $\gamma_t = \W_{m-1}$,  after applying the identification $\sj$ defined in \eqref{sj}.
\end{Pro}
\begin{proof}  Differentiating \eqref{Sym0} with respect to time gives
$$\Gamma_t = \left.\Psi^{-1} \dfrac{\partial V_m}{\partial \lambda} \Psi \right\vert_{\lambda=0}.$$
Using \eqref{Vmex},
$$\left.\dfrac{\partial V_m}{\partial \lambda} \right\vert_{\lambda=0} = -P_{m-1} = f_{m-1} \me_0 + g_{m-1} \me_1 + h_{m-1} \me_2.$$
Then, using the notation of \eqref{jangle} with $\lambda=0$,
$\Gamma_t =f_{m-1} \langle\me_0\rangle + g_{m-1} \langle\me_1\rangle+ h_{m-1} \langle\me_2\rangle$,
and the conclusion follows by applying $\sj$.

\end{proof}

\subsection{Gauge Transformations} \label{apgauge}

In this section we calculate the effect of gauge transformations on solutions of the AKNS system for $m$th order NLS.
What we find will ultimately let us say what evolution equation is satisfied by a curve produced by using
these AKNS solutions in the Sym formula.

The gauge transformation we consider replaces the AKNS solution $\vpsi$ with
$$\hatpsi = \exp( -(\alpha x+\beta t) \me_0) \vpsi = \dfrac{\ri}{2}\begin{pmatrix} e^{\alpha x+\beta t} & 0 \\ 0 & -e^{-(\alpha x+\beta t)} \end{pmatrix} \vpsi$$
for $\alpha,\beta\in \bbR$.
In particular, we will show that $\hatpsi$ satisfies a modified AKNS system, for a modified potential
and a shifted spectral parameter.

It is easy to see the modification of the potential and the shift in $\lambda$ by computing the $x$-derivative
of $\hatpsi$.  Recall from \eqref{AKNSmx} and \eqref{defofU} that $\vpsi_x = U \vpsi = (\lambda \me_0 + Q)\vpsi$.
Then
\eqnn{
\hatpsi_x &= -\alpha \me_0 \hatpsi + \left(\Ad_{\exp( -(\alpha x+\beta t) \me_0)} U\right) \hatpsi\\
&= ((\lambda-\alpha) \me_0 + \Ad_{\exp( -(\alpha x+\beta t) \me_0)} Q) \hatpsi = \left( (\lambda-\alpha) \me_0 + \Qhat \right) \hatpsi,
}
where we define
$$\Qhat:= \dfrac{\ri}2\begin{pmatrix} 0 & \qhat \\ \overline{\qhat} & 0 \end{pmatrix}\quad\text{and}\quad \qhat(x,t) := e^{\ri(\alpha x+\beta t)} q(x,t).$$
More generally, in the rest of this subsection a `hat' accent will denote the result of replacing $q$ by $\qhat$, and for a matrix $M$
we will let $\{M \} = \Ad_{\exp( -(\alpha x+\beta t) \me_0)} M$.  Sometimes these coincide, as with $\{ P_0 \} = \Phat_0 = \me_0$ (which is a constant
anyway) and $\{ Q\} = \Qhat$.

Using $\vpsi_t = V_m \vpsi$, we can similarly compute that
\begin{equation}\label{hatpsitime}
\hatpsi_t = (-\beta \me_0 + \{ V_m \} ) \hatpsi.
\end{equation}
Recall that $V_m$ is an $m$th degree polynomial in $\lambda$, with matrix-valued coefficients
depending on $q$ and its derivatives.
It turns out that $\{V_m(q,\lambda)\}$ can be expressed as a linear combination of $V_j(\qhat, \lambda-\alpha)$ for $j\le m$.

%
\begin{Pro}\label{Vgaugeprop}
\begin{equation}\label{Vconjexp}
\{ V_m(q,\lambda)\} = \sum_{k=0}^m (-\alpha)^k \binom{m}k V_{m-k}(\qhat, \lambda-\alpha),
\end{equation}
where we let $V_0=P_0$.  In particular, if $\beta=(-\alpha)^m$, then from \eqref{hatpsitime} we have
$$\hatpsi_t = \sum_{k=0}^{m-1} (-\alpha)^k \binom{m}k V_{m-k}(\qhat,\lambda-\alpha) \hatpsi,$$
and it follows that
\begin{equation}\label{qhatty}
\qhat_t = \displaystyle\sum_{k=0}^{m-1} (-\alpha)^k \binom{m}k F_{m-k}[\qhat],
\end{equation}
{where $F_{j}[\qhat]$ denotes the RHS of $\NLS_j$, as in \eqref{NLSm1}}.
\end{Pro}

\begin{Cor}\label{modCor} Let $q$ satisfy $\NLS_m$ \eqref{NLSm1} and let $\Gamma$ be obtained from the Sym formula \eqref{Symf}
evaluated at $\lambda_0 \in \bbR$.  Then $\gamma = \sj\circ \Gamma$  satisfies
$$\gamma_t = \sum_{k=0}^{m-1} (-\lambda_0)^k \binom{m}k \W_{m-1-k}.$$
\end{Cor}
\begin{proof}
As in the proposition just above, we take $\beta=(-\alpha)^m$ and $\alpha=\lambda_0$.
We let
$$\widehat{\bPsi} = \exp(-(\alpha x+\beta t)\me_0) \bPsi.$$
Then, using the change of variable $\widehat\lambda = \lambda -\alpha$, we see that
$$\Gamma = \left.\bPsi^{-1} \dfrac{\partial \bPsi}{\partial \lambda} \right\vert_{\lambda = \lambda_0}
    = \left.\widehat\bPsi^{-1} \dfrac{\partial \widehat\bPsi}{\partial \widehat\lambda} \right\vert_{\widehat\lambda = 0}.
$$
In other words, the same evolving curve arises from the Sym transformation {\em evaluated at zero} with potential $\qhat$.
Because Prop.~\ref{allSym0} shows that, under the Sym formula evaluated at zero, a solution of $\NLS_m$ \eqref{NLSm1}
generates a solution of $\VFE_{m-1}$ \eqref{filn}, then $\gamma$ evolves by the linear combination of VFE flows corresponding
to \eqref{qhatty}.
\end{proof}

To prove Prop.~\ref{Vgaugeprop} we first need to express $\{ P_m \}$ in terms of the $\Phat_j$ for $j\le m$.
\begin{Lem}
\begin{equation}\label{Pconjexp}
\{ P_m \} = \sum_{k=0}^{m} (-\alpha)^k \binom{m-1}k \Phat_{m-k},
\end{equation}
with the convention that $\binom{j}0 = 1$ for any integer $j$, and $\binom{m-1}m = 0$ for $m>0$.
\end{Lem}
\begin{proof}
Recall from \eqref{mittelj} that the $P_m$ satisfy the recurrence relation $[ P_{m+1}, \me_0] = \di P_m - [Q,P_m]$.
Conjugating by $\exp( -(\alpha x+\beta t) \me_0)$ gives
$$[ \{ P_{m+1} \}, \me_0] = \{  \di P_m - [Q,P_m] \} = \{  \di P_m \} - [ \Qhat, \{ P_m\} ].$$
But $ \di \{ P_m \} =  \di Ad_{\exp(-(\alpha x + \beta t) \me_0)} P_m = [-\alpha \me_0, \{ P_m\} ] + \{  \di P_m \}$, so
\begin{equation}\label{conjrel}
[ \{ P_{m+1} + a P_m \}, \me_0 ] = ( \di - \ad_\Qhat) \{ P_m \}.
\end{equation}

Let $S_m$ denote the right-hand side of \eqref{Pconjexp}; we will show that the $S_m$
satisfy the same recurrence relation \eqref{conjrel} {as satisfied by} the $\{P_m\}$.
The left-hand side of this relation is
\begin{align*}
[S_{m+1} + a S_m , \me_0] =&\sum_{k=0}^{m+1} (-\alpha)^k \binom{m}k [ \Phat_{m+1-k}, \me_0] + a \sum_{k=0}^m (-\alpha)^k \binom{m-1}k [ \Phat_{m-k}, \me_0] \\
=& [\Phat_{m+1}, \me_0] + \sum_{k=1}^{m+1} (-\alpha)^k \left( \binom{m}k - \binom{m-1}{k-1} \right) [ \Phat_{m-k+1}, \me_0] \\
=& \sum_{k=0}^m (-\alpha)^k  \binom{m-1}k [ \Phat_{m-k+1}, \me_0].
\end{align*}
The right-hand side of the relation is
$$
(\di - \ad_{\Qhat}) S_m = \sum_{k=0}^m (-\alpha)^k \binom{m-1}k ( \di - \ad_\Qhat) \Phat_{m-k}
= \sum_{k=0}^m (-\alpha)^k \binom{m-1}k[ \Phat_{m-k+1}, \me_0],
$$
where on the right we used the fact that the $\Phat_k$ satisfy the same recurrence relations, but with $\Qhat$ in place of $Q$.

Since the $S_m$ satisfy the same recurrence as the $\{P_m\}$, these sequences
of matrices coincide up to adding constant multiples of $\me_0$.  But neither sequence (after $m=0$) includes
constant multiplies of $\me_0$, so we conclude that $\{P_m\} = S_m$ for all $m\ge 0$.
\end{proof}

\begin{proof}[Proof of Prop.~\ref{Vgaugeprop}]
Using \eqref{Pconjexp}, we compute the left-hand side of \eqref{Vconjexp} as
\begin{align}
\{ V_m(q,\lambda)\} =\sum_{k=0}^m (-\lambda)^k \{ P_{m-k}\} &= \sum_{k=0}^m (-\lambda)^k \sum_{j=0}^{m-k} (-\alpha)^j \binom{m-k-1}j \Phat_{m-k-j} \notag \\
		&= \sum_{\ell=0}^m (-1)^\ell \left(\sum_{k=0}^\ell \alpha^{\ell-k} \lambda^k  \binom{m-k-1}{\ell-k} \right)\Phat_{m-\ell}, \label{VmLHS}
\end{align}
letting $\ell=j+k$ and reversing the order of summation.
We compute the right-hand side of \eqref{Vconjexp} as
\begin{equation}\nonumber
\sum_{k=0}^m (-\alpha)^k \binom{m}k \sum_{j=0}^{m-k} (a-\lambda)^j \Phat_{m-k-j}
= \sum_{\ell=0}^m (-1)^\ell \left( \sum_{j=0}^\ell \binom{m}{\ell-j} \alpha^{\ell-j} (\lambda-\alpha)^j\right) \Phat_{m-\ell}.
\end{equation}
We expand the coefficient in parentheses using the
binomial theorem:
\begin{align*}
\sum_{j=0}^\ell \binom{m}{\ell-j} \alpha^{\ell-j} \sum_{k=0}^j \binom{j}k \lambda^k (-\alpha)^{j-k}
&= \sum_{k=0}^\ell \alpha^{\ell-k} \lambda^k \sum_{j=k}^\ell (-1)^{j-k} \binom{m}{\ell-j} \binom{j}k \\
&= \sum_{k=0}^\ell \alpha^{\ell-k} \lambda^k \sum_{r=0}^{\ell-k} (-1)^r\binom{m}{\ell-k-r} \binom{k+r}{k},
\end{align*}
where $r=j-k$. By comparing this with the coefficient in parentheses in \eqref{VmLHS},
we see that the result follows from the binomial coefficient identity
$$\sum_{r=0}^{\ell-k} (-1)^r\binom{m}{\ell-k-r} \binom{k+r}{k} = \binom{m-k-1}{\ell-k}$$
for $0\le k \le \ell \le m$.  This identity is easily proven by induction on $m$.
\end{proof}


\subsection{Finite-Gap Solutions}
\label{Apfg}
Let $\bPsi(x,t;\lambda)$, $\lambda \in \C$ be a fundamental matrix solution of the AKNS system
for $\NLS_m$: 
\begin{equation}\label{NLSm}
\vpsi_x = U \vpsi,\qquad \vpsi_t = V_m \vpsi.
\end{equation}
Suppose $|q(x,t)|$ is bounded in $x$ and $t$.  Based on the leading-order parts of $U$ and $V_m$, the asymptotic
behavior of $\bPsi$ as $\lambda \to \infty$ is given by
$$\bPsi \sim \exp\left( (\lambda x -(-\lambda)^m t) \me_0\right),$$
up to right multiplication by a matrix independent of $x$ and $t$.  The algebro-geometric solutions of $\NLS_m$ are constructed
by supposing that there is matrix solution of the form
\begin{equation}\label{Phiasy}
\bPsi =\left( I + \sum_{k=0}^\infty \lambda^{-k}\bPsi_k(x,t) \right)
\exp\left( (\lambda x -(-\lambda)^m t) \me_0\right).
\end{equation}
The matrix $\bPsi$ is assembled from the two branches of a {\em Baker-Akhiezer} function on a compact hyperelliptic Riemann surface $\Sigma$.
(The following development is a summary of more detailed discussions in \cite{BBEIM} and \cite{CI05}.)

Let $\Sigma$ be such a Riemann surface of genus $g$, with $2g+2$ distinct finite
branch points arranged in complex conjugate pairs.  Thus, $\Sigma$ is defined by
$$\mu^2 = \sum_{j=1}^{g+1} (\lambda-\lambda_j)(\lambda-\overline{\lambda_j}).$$
Let $\pi: (\lambda, \mu) \mapsto \lambda$ be the branched double cover from $\Sigma$ to the Riemann sphere.
The inverse image under $\pi$ of $\lambda=\infty$ is a pair of points $\infty_\pm$ characterized by $\mu/\lambda^{g+1} \sim \pm 1$.
In general, a Baker-Akhiezer function is meromorphic in the finite part $\Sigma -\{\infty_+, \infty_-\}$, with prescribed essential singularities at $\infty_\pm$.  To construct $\bPsi$, we want a $\C^2$-valued function $\vpsi$ with the following behavior at $\infty_\pm$:
\begin{align*}
\vpsi(P,x,t) &\sim \exp(-\tfrac12 \ri (\lambda x -(-\lambda)^m t))\begin{bmatrix}1 \\ 0\end{bmatrix} \quad \text{ as }P\to \infty_-,\\
\vpsi(P,x,t) &\sim \exp(\tfrac12 \ri (\lambda x -(-\lambda)^m t))\begin{bmatrix}0 \\ 1\end{bmatrix} \quad \text{ as } P\to \infty_+,\\
\end{align*}
corresponding to the first and second columns of $\bPsi$ respectively.

\begin{Pro}[Prop.~2.2 in \cite{CI05}]\label{closureprop} Let $\D$ be a positive divisor of degree $g+1$ on $\Sigma$, such that $\D -\infty_+ -\infty_-$ is not linearly equivalent to a positive divisor.  Then there exists a unique $\C^2$-valued
function $\vpsi(x,t;P)$ with the above asymptotic behavior as $P \to \infty_\pm$, and which is
meromorphic in the finite part of $\Sigma$, with pole divisor contained in $\D$.
\end{Pro}

We will give a formula for the Baker-Akhiezer function below.  Its asymptotic behavior is attained using a pair of Abelian integrals on $\Sigma$ defined
as follows.  Fix a basis $(a_j, b_j)$, $j=1, \ldots, g$  of homology cycles on $\Sigma$ such that $\pi(a_j)$
encloses only the branch points $\lambda_j$ and $\overline{\lambda_j}$, and which
has the standard intersection pairings $a_j \cdot a_k =  b_j \cdot b_k =0$, $a_j \cdot b_k =\delta_{jk}$.
(For the sake of definiteness, let $\overline{\lambda_{g+1}}$ be the basepoint not enclosed by any of the $a$- or $b$-cycles.)  Let $\omega_j$ be holomorphic differentials on $\Sigma$ satisfying
$\int_{a_k} \omega_j = 2\pi \ri \delta_{jk}.$
Let $d\Omega_1$ and $d\Omega_m$ be meromorphic differentials on $\Sigma$ that have zero $a$-periods and
satisfy
$$d\Omega_1 \sim \pm \tfrac12 d\lambda, \quad d\Omega_m \sim \pm \tfrac12 m (-\lambda)^{m-1} d\lambda \qquad \text{as } \lambda \to \infty_\pm.$$
(These are constructed by beginning with the meromorphic differentials $(\lambda^{g+1}/\mu)d\lambda$ and
$(\lambda^{g+m}/\mu)d\lambda$, subtracting off multiples of the $\omega_j$'s to cancel the $a$-periods, and scaling with
appropriate constants.)
Then, using $\overline{\lambda_{g+1}}$ as basepoint, we define the integrals
$$\Omega_1(P) = \int^P d\Omega_1,\qquad \Omega_m(P) =\int^P d\Omega_m.$$
By our choice of basepoint, and because
 because $\imath^* d\Omega_1 = -d\Omega_1$ and $\imath^* d\Omega_m = -d\Omega_m$ (where $\imath$  is the sheet interchange automorphism of $\Sigma$), we have
$$\Omega_1(P)  = \pm\tfrac12(\lambda-E)+O(\lambda^{-1}),\quad \Omega_m(P) = \pm\tfrac12 (N-(-\lambda)^m)+O(\lambda^{-1})
$$
as $P \to \infty_\pm$, for some real constants $E$ and $N$.

The formulas for the components of the Baker-Akhiezer function are then
\eq{
\psi_{1} (x, t; P) &= \exp\left({\ri( (\Omega_1(P) - \tfrac{E}2)x+(\Omega_m(P)+\tfrac{N}2)t)}\right)
\\ &\times
 \frac{ \theta(\A(P) + \ri \vV x + \ri \vW t  -\vD) \theta(\vD)}
{\theta(\ri \vV x + \ri \vW t  - \vD) \theta(\A(P) -\vD)} g_+(P)
\\ \\
\psi_{2}(x,t;P) &=\exp\left( \ri (\Omega_1(P)+\tfrac{E}2)x +(\Omega_m(P)-\tfrac{N}2)t)\right) \\
 &\times \frac{\theta(\A(P) +\ri \vV x +\ri \vW t -\vD -\vr)\theta(\vD)}
{\theta(\ri\vV x+\ri\vW t-\vD) \theta(\A(P)-\vD-\vr)} g_-(P),
}{Baker}
where
\begin{itemize}
\item $\theta$ is the Riemann theta function determined by the period matrix
$B_{jk} = \int_{b_k} \omega_j$:
$$\theta(\vz) = \sum_{n \in \ZZ^g} \exp \langle \vz + \tfrac12 Bn,n \rangle, \qquad \vz \in \C^g.$$
\item $\A$ is the Abel map $P \mapsto \int_{\infty_-}^P \vomega$.
(The paths of integration for $\A(P)$ and $\Omega_1(P)$ (or $\Omega_m(P)$) are
yoked, in the sense that they differ by a fixed path from $\infty_-$ to $\overline{\lambda_{g+1}}$
that avoids the $a$- and $b$-cycles.  This makes $\psi_1, \psi_2$ path-independent.)
\item $\vV, \vW$ are vectors with components $V_j = \int_{b_j} d\Omega_1$ and $W_j =\int_{b_j} d\Omega_m$.
\item We let $\D_+$ and $\D_-$ be positive divisors linearly equivalent to $\D - \infty_+$ and $\D-\infty_-$ respectively.
Then $\vD = \A(\D_+) + \vK$, where $\vK$ is the vector of Riemann constants for $\Sigma$
(see \cite{BBEIM}, \S2.7), and $\vr = \int_{\infty_-}^{\infty_+} \vomega$, so that $\A(\D_-) = \vD + \vr - \vK$.
\item $g_\pm$ are meromorphic functions on $\Sigma$ with poles on $\D$ and zeros on $\D_\pm + \infty_\pm$,
normalized so that $g_+(\infty_-) = g_-(\infty_+) = 1$.
\end{itemize}

We assemble a fundamental matrix solution for \eqref{NLSm} whose columns are given by
the Baker-Akhiezer function:
$$\bPsi(x,t;\lambda) = \left( \vpsi(P^-) \ \vpsi(P^+ \right) ),$$
where $P^+$ and $P^-$ are inverse images of $\lambda$ under $\pi$ that
are near $\infty_+$ and $\infty_-$ when $\lambda$ is near $\infty$.
It then follows that $\bPsi$ has the asymptotic behavior given by \eqref{Phiasy}.
Then the argument of pp. 93--94 in \cite{BBEIM} shows that $\bPsi$ actually solves \eqref{NLSm},
for the potential given by
\begin{equation}\label{finitegapq}
q(x,t) = A \exp(-\ri E x+\ri N t) \dfrac{\theta(\ri \vV x +\ri \vW t - \vD +\vr)}{\theta( \ri \vV x +\ri \vW t - \vD )},
\end{equation}
where $A = \dfrac{2\theta(\vD)}{\alpha \theta(\vD-\vr)}$ and {$\alpha \in \bbR$ is such that} $g_+ \sim (\alpha \lambda)^{-1}$ near $\infty_+$.

\subsection*{Closure Conditions}

Now we consider the opposite construction: given solution $q(x,t)$ of $\NLS_m$ which is $L$-periodic in $x$,  at what value $\lambda=\lambda_0$ does the
Sym formula \eqref{Symf} yield a smooth closed curve of length $L$?   We will confine our attention to algebro-geometric solutions.

First, note that for the solution \eqref{finitegapq} to be periodic, it is necessary that the
numbers $(V_1, V_2, \ldots, V_g, E)$ be integer multiples of $2\pi/L$.
The general criteria for closure of the curve resulting from the Sym formula was given
by Grinevich and Schmidt \cite{GS}; the following is a specialization to finite-gap solutions:

\begin{Pro}[Prop.~2.4 in \cite{CI05}]\label{prop2} Let $q(x,t)$ be a finite-gap potential of period $L$.  Then
the curve given by the Sym formula \eqref{Symf} is smoothly closed of length $L$
if and only if (a) $\exp( \ri \Omega_1(P) L) = \pm 1$, and (b) $d\Omega_1(P)=0$, where $\pi(P)=\lambda_0$.
\end{Pro}
In other words,  the chosen value $\lambda=\lambda_0$ in the Sym {formula} \eqref{Symf} must be a periodic or antiperiodic point of the Floquet spectrum of $q$, and a zero of the quasimomentum differential $d\Omega_1$ (see \cite{CI05}, section 2.3 for more details).

\begin{remark}\label{quasiclosure}
We can extend Prop.~\ref{prop2} to potentials for which only the quotient of $\theta$-functions in \eqref{finitegapq} is periodic in $x$
(i.e., neglecting the exponential factor).  For this, it is necessary and sufficient that the $V_j$ be rationally related.
Configurations of branch points that satisfy this rationality condition can be obtained by successive isoperiodic deformations
starting with the spectrum of the multiply-covered circle \cite{CI07}.
Once there is a real number $L$ such that such that $V_j = 2\pi m_j/L$ for some integers $m_j$, we say that
the finite-gap solution $q$ is {\em L-quasiperiodic}, since
\begin{equation}\label{quasiq}
q(x+L, t) = e^{-\ri E L} q(x,t).
\end{equation}
However, when we apply a shift transformation $\sigma(\lambda) = \lambda - s$ for $s \in \bbR$, the branch points
can be translated to the right or left to arrange that $E =0$.  (For,
the pullback of $d\Omega_1$ under $\sigma$ is the quasimomentum differential for the original hyperelliptic curve; thus,
the value of the $V_j$ is unchanged, and $E$ is replaced by $E-s$.)  Now using \eqref{finitegapq}
with the shifted configuration gives an $L$-periodic potential $\tilde q$.
The {\em periodic points} of the Floquet spectrum of $\tilde q$ consists of $\lambda$-values for which there are non-trivial eigenfunctions of period (or anti-period) $L$,
corresponding to condition (a) above.  Likewise, we say that $\lambda$ is a periodic point for $q$ if its
image under the shift $\lambda \mapsto \lambda - E$ is a periodic point for $\tilde q$.
In particular, the closure conditions (a) and (b) are preserved under the shift.  Moreover, as in Corollary
\ref{modCor}, the curve constructed using $q$ in the Sym formula evaluated at $\lambda_0$ will be identical, at $t=0$,
to that construct using $\tilde q$ in the Sym formula evaluated at $\sigma(\lambda_0)$.
\end{remark}
\end{appendices}


\end{document}